\documentclass[a4paper,11pt]{article} 
   
\usepackage[english]{babel}
\usepackage[a4paper]{geometry}
\usepackage{enumerate}  
\usepackage{amsmath}
\usepackage{graphicx}
\usepackage{amsthm}
\usepackage{amssymb}
\usepackage{hyperref} 

\def\bR{\mathbb{R}}

\def\cN{\mathcal{N}}
\def\cQ{\mathcal{Q}}
\def\cF{\mathcal{F}}

\def\cE{\mathcal{E}}
\def\cK{\mathcal{K}}
\def\cH{\mathcal{H}}

\def\ph{\varphi}

\def\cV{\mathcal{V}}

\def\cU{\mathcal{U}}
\def\cS{\mathcal{S}}

\def\be{\begin{equation}}
\def\ee{\end{equation}}

\def\rhs{r.h.s.\ }
\def\lhs{l.h.s.\ }
\def\eg{e.g.\ }

\def\wrt{w.r.t.\ }

\newtheorem{theorem}{Theorem}   [section]

\newtheorem{prop}[theorem]{Proposition}
\newtheorem{lemma}[theorem]{Lemma}

\numberwithin{equation}{section}

\begin{document}
\title{Derivation of the Gross-Pitaevskii Dynamics through Renormalized Excitation Number Operators}
\author{Christian Brennecke\thanks{Institute for Applied Mathematics, University of Bonn, Endenicher Allee 60, 53115 Bonn, Germany}  \and  Wilhelm Kroschinsky\footnotemark[1] }

\maketitle

\begin{abstract}
We revisit the time evolution of initially trapped Bose-Einstein condensates in the Gross-Pitaevskii regime. We show that the system continues to exhibit BEC once the trap has been released and that the dynamics of the condensate is described by the time-dependent Gross-Pitaevskii equation. Like the recent work \cite{BS}, we obtain optimal bounds on the number of excitations orthogonal to the condensate state. In contrast to \cite{BS}, however, whose main strategy consists of controlling the number of excitations with regards to a suitable fluctuation dynamics $t\mapsto e^{-B_t} e^{-iH_Nt}$ with renormalized generator, our proof is based on controlling renormalized excitation number operators directly with regards to the Schr\"odinger dynamics $t\mapsto e^{-iH_Nt}$. 
\end{abstract}

\section{Introduction and Main Result} \label{sec:intro}

The mathematical analysis of spectral and dynamical properties of dilute Bose gases has seen tremendous progress in the past decades after the first experimental observation of Bose-Einstein condensates in trapped atomic gases \cite{CW95, K95}. In this work, we model such experimental setups by considering $N$ bosons moving in $\bR^3$ with energies described by 
		\be \label{eq:Htrap} 
		H^\text{trap}_N = \sum_{j=1}^N \left( -\Delta_{x_j} + V_\text{ext} (x_j) \right) + \sum_{1\leq i<j\leq N} N^2 V(N(x_i -x_j)), 
		\ee 
which acts on a dense subspace of $ L^2_s(\bR^{3N})$, the subspace of $ L^2(\bR^{3N})$ that consists of wave functions that are invariant under permutation of the particle coordinates. We assume the two body interaction $V\in L^1(\bR^3) $ to be pointwise non-negative, radially symmetric and of compact support. The trapping potential $V_\text{ext}\in L_{\text {loc}}^\infty(\bR^3)  $ is assumed to be locally bounded and to satisfy $ \lim_{|x|\to\infty} V_{\text{ext}}(x) =\infty$. 

The scaling $ V_N =  N^2V(N.)$ characterizes the Gross-Pitaevskii scaling which can be understood as a joint thermodynamic and low density limit as $N\to\infty$ (see \eg \cite{LSSY} for a detailed introduction). It ensures that the rescaled potential $V_N$ has a scattering length $ \mathfrak{a} (V_N) = N^{-1} \mathfrak{a} (V)$ of order $O(N^{-1})$ so that both the kinetic and potential energies in \eqref{eq:Htrap} are typically of size $O(N)$ \wrt low energy states. In fact, it is well-known \cite{LSY} that the scattering length $ \mathfrak{a} \equiv \mathfrak{a}(V)$ completely characterizes the influence of the interaction on the leading order contribution to the ground state energy $ E_N = \inf \text{spec} ( H^\text{trap}_N)$:  
		\be \label{eq:energyld} \lim_{N\to\infty} \frac{E_N}{N} = \inf_{\ph\in L^2(\bR^3)} \cE_{\text{GP}}^{\text{trap}}(\ph) \equiv e^\text{trap}_\text{GP}. \ee
Here, $ \cE^{\text{trap}}_{\text{GP}}$ denotes the Gross-Pitaevskii energy functional defined by
		\be\label{def:GPfunc}    \cE_{\text{GP}}^{\text{trap}}(\ph) = \int_{\bR^3} dx\, \Big(  \vert \nabla \ph(x) \vert^2 + V_\text{ext}(x) \vert \ph(x)\vert^2 + 4\pi\mathfrak{a} \vert \ph(x)\vert^4 \Big)   \ee
and the scattering length $\mathfrak{a}$ of the potential $V$ is characterized by
		\be \label{def:a} \mathfrak{a} = \frac1{8\pi}\inf \bigg\{ \int_{\bR^3} dx\,\big( 2|\nabla f(x)|^2 + V(x)|f(x)|^2\big): \lim_{|x|\to\infty}f(x)=1 \bigg\}.\ee
By standard variational arguments, the functional \eqref{def:GPfunc} admits a unique positive, normalized minimizer, denoted in the sequel by $ \ph_{\text{GP}}$, and it turns out that the normalized ground state $ \psi_N$ of $ H^\text{trap}_N$ exhibits complete Bose-Einstein condensation into $ \ph_{\text{GP}}$: if $ \gamma_N^{(1)} =\text{tr}_{2,\ldots,N} |\psi_N\rangle\langle\psi_N|$ denotes the one-particle reduced density of $\psi_N$, then \cite{LS1}
		\be \label{eq:BEC} \lim_{N\to\infty} \langle  \ph_{\text{GP}}, \gamma_N^{(1)}\ph_{\text{GP}}\rangle = 1. \ee
Physically, the identity \eqref{eq:BEC} means that the fraction of particles occupying the condensate state $ \ph_{\text{GP}}$ tends to one in the limit $N\to \infty$. Mathematically, it is equivalent to the convergence of $ \gamma_N^{(1)}$ to the rank-one projection $|\ph_{\text{GP}}\rangle\langle \ph_{\text{GP}}| $ in trace class which implies that one body observables are asymptotically completely determined by $\ph_\text{GP}$. 

It should be noted that the convergence in \eqref{eq:BEC} holds true more generally for approximate ground states $ \psi_N$ that satisfy $ N^{-1}\langle \psi_N,H^\text{trap}_N \psi_N\rangle\leq e^{\text{trap}}_{\text{GP}}  + o(1)$ for an error $o(1)\to 0$ as $N\to\infty$. This has been proved in \cite{LS2} and later been revisited with different tools in \cite{NRS}. Moreover, very recent developments have lead to a significantly improved quantitative understanding of \eqref{eq:energyld} and \eqref{eq:BEC}: generalizing previously obtained results in the translation invariant setting \cite{BBCS1, BBCS2, BBCS3, BBCS4, ABS}, the works \cite{NNRT, BSS1} determine the optimal convergence rates in \eqref{eq:energyld} and \eqref{eq:BEC} while \cite{NT, BSS2} go a step further and determine the low energy excitation spectrum of $ H^\text{trap}_N$ up to errors $o(1)\to 0$ that vanish in the limit $ N\to\infty$. In particular, the main results of \cite{NT, BSS2} imply that the ground state and elementary excitation energies of $ H^\text{trap}_N$ depend on the interaction up to second order only through its scattering length, in accordance with Bogoliubov's predictions \cite{Bog}. It is remarkable that this even remains true up to the third order contribution to the ground state energy of size $ \log N/N$, as recently shown for translation invariant systems in \cite{COSAS}. 

In view of experimental observations of Bose-Einstein condensates, it is natural to study the dynamics of initially trapped Bose-Einstein condensates and to ask whether the system continues to exhibit BEC once the trap is released. Based on the preceding remarks, it is particularly interesting to consider an approximate ground state $ \psi_N$ of $ H_N^{\text{trap}}$ and to analyze its time evolution after releasing the trap $ V_{\text{ext}}$. We model this situation by studying the Schr\"odinger dynamics 
		\[ t\mapsto \psi_{N,t}=e^{-i H_N t}\psi_N \] 
generated by the translation invariant Hamiltonian $H_N$ which is given by
		\be \label{eq:H}  
		H_N = \sum_{j=1}^N -\Delta_{x_j}   + \sum_{1\leq i<j\leq N} N^2 V(N(x_i -x_j)).
		\ee 
As in the spectral setting, it turns out that also the dynamics is determined to leading order by the Gross-Pitaevskii theory: if $ \gamma_{N,t}^{(1)} =\text{tr}_{2,\ldots,N} |\psi_{N,t}\rangle\langle\psi_{N,t}|$ denotes the reduced one-particle density with regards to the evolved state $ \psi_{N,t}$, then \cite{ESY1, ESY2, ESY3, ESY4}
		\be \label{eq:BECdyn} \lim_{N\to\infty} \langle  \ph_{t}, \gamma_{N,t}^{(1)}\ph_{t}\rangle = 1 \ee
for all $t\in \bR$, where $ t\mapsto \ph_t$ solves the time-dependent Gross-Pitaevskii equation
		\be \label{eq:GPdyn}   i\partial_t \ph_t = -\Delta \ph_t + 8\pi \mathfrak{a} |\ph_t|^2 \ph_t. \ee
Like in the spectral setting, the convergence \eqref{eq:BECdyn} can be quantified with an explicit rate as shown first in \cite{P}, later in a Fock space setting in \cite{BDS} and, generalizing the main strategy of \cite{BDS}, with optimal convergence rate in \cite{BS}. Moreover, quite recently the dynamical understanding has been further improved in \cite{COS} which provides a quasi-free approximation of the many body dynamics $ t\mapsto \psi_{N,t}$ with regards to the $ L_s^2(\bR^{3N})$-norm. Comparable norm approximations were previously only available in scaling regimes that interpolate between the mean field and Gross-Pitaevskii regimes, but excluding the latter; for more details on this see \eg \cite{GMM1, GMM2, C, LNS, MPP, BCS, NN1, NN2, K, BNNS, BPaPS, BPePS}.

Although the norm approximation provided in \cite{COS} is of independent interest, the results of \cite{COS} unfortunately do not suffice, yet, to effectively compute important observables such as the time evolved number of excitations orthogonal to the condensate $\ph_t$ or their energy in terms of the quasi-free dynamics, up to errors that vanish in the limit $N\to \infty$ (see also Remark 5) below for a related comment). This likely requires stronger a priori estimates on the full many body evolution $ t\mapsto \psi_{N,t}$ than those proved in \cite{BS}, which are an important ingredient in the proof of \cite{COS}. Since the arguments of \cite{BS} are rather involved, it thus seems first of all worthwhile to revisit and streamline its proof. This is our main motivation and, inspired by recent simplifications in the spectral setting \cite{BFS, H, HST, Br, BBCO}, we provide a novel and, compared to previous derivations, substantially shorter proof of \eqref{eq:BECdyn}. To this end, we combine some algebraic ideas as introduced in \cite{Br} with some of the main ideas of \cite{BS}. Our main result is as follows.
 
\begin{theorem}\label{thm:main}
Let $ V\in L^1(\bR^3)$ be non-negative, radial and of compact support, and let $V_\emph{ext}\in L_{\emph{loc}}^\infty(\bR^3)  $ be such that $ \lim_{|x|\to\infty} V_{\emph{ext}}(x) =\infty$. Let $\psi_N\in L^2_s(\bR^{3N})$ be normalized with one-particle reduced density $\gamma^{(1)}_N $ and assume that     
		\begin{equation}\label{eq:ass}
		\begin{split} 
		o_1(1) & = \big| N^{-1}\langle \psi_N, H^\emph{trap}_N \psi_N \rangle - e^\emph{trap}_\emph{GP} \big| \to 0, \;\;\;\;o_2(1) = 1 - \langle \ph_\emph{GP} , \gamma^{(1)}_N \ph_\emph{GP} \rangle  \to 0,
		\end{split} 
		\end{equation}
in the limit $N \to \infty$, where $\ph_\emph{GP}$ denotes the unique positive, normalized minimizer of the Gross-Pitaevskii functional \eqref{def:GPfunc}. Assume furthermore that $\ph_\emph{GP} \in H^4 (\bR^3)$. 

Then, if $t\mapsto \psi_{N,t} = e^{-i H_N t} \psi_N$ denotes the Schr\"odinger evolution and $\gamma_{N,t}^{(1)}$ its reduced one-particle density, there are constants $C,c > 0$ such that 
		\begin{equation}\label{eq:main} 
		1 - \langle \ph_t , \gamma^{(1)}_{N,t} \ph_t \rangle \leq C \left[ o_1(1) + o_2(1) + N^{-1} \right] \exp \, (c \exp \, (c|t|)) 
		\end{equation}
for all $t \in \bR$, where $t\mapsto \ph_t$ denotes the solution of the time dependent Gross-Pitaevskii equation \eqref{eq:GPdyn} with initial data $\ph_{|t=0} = \ph_\emph{GP}$.
\end{theorem}

\noindent \textbf{Remarks:} 
\begin{enumerate}[1)]
\item Theorem \ref{thm:main} was previously shown in \cite[Theorem 1.1]{BS} under the slightly stronger assumption $ V\in L^3(\bR^3)$. Our main contribution is a novel and short proof, valid for $ V\in L^1(\bR^3)$, which is outlined in detail in Section \ref{sec:outline}. 
\item Under suitable conditions on $ V$ and $ V_{\text{ext}}$, the main results of \cite{NNRT, BSS1, NT, BSS2} imply that the assumptions \eqref{eq:ass} are satisfied for low energy states with an explicit rate. Applying these results to the ground state of $ H^\text{trap}_N$, one finds that $ o_1(1)=O(N^{-1})$ and $ o_2(1)=O(N^{-1})$ so that the overall convergence rate in \eqref{eq:main} is of order $ O(N^{-1})$. The quasi-free approximation obtained in \cite{COS} implies that this rate is optimal in $N$.
\item As mentioned earlier, we adapt recent ideas from \cite{Br}, which analyzes the spectrum of Bose gases for translation invariant systems, to the dynamical setting. To illustrate further the usefulness of the method, in particular in the context of Theorem \ref{thm:main}, we sketch in Appendix \ref{app:smallV} an elementary proof of \eqref{eq:ass} with optimal rate for the ground state $\psi_N$ of $H_N^{\text{trap}}$ if $ \|V\|_1$ is sufficiently small. This is analogous to the main result of \cite{BBCS1} in the translation invariant setting. Note that \cite{NNRT} provides a different proof for $V\in L^1(\bR^3)$ under the milder assumption that $ \mathfrak{a}$ is small and that \cite{BSS1} proves a similar result for $V\in L^3(\bR^3)$ without smallness assumption on $ \mathfrak{a}$. Compared to Appendix \ref{app:smallV}, these results have required, however, substantially more work.
\item As already pointed out in \cite{BS}, the assumption that $\ph_\text{GP} \in H^4 (\bR^3)$ follows \eg from suitable growth and regularity assumptions on $ V_{\text{ext}}$, based on the Euler-Lagrange equation for $ \ph_{\text{GP}}$ and elliptic regularity arguments. Since we are not aware of a precise condition on $V_{\text{ext}}$ that guarantees the improved regularity of $\ph_\text{GP}$, we explicitly assume $\ph_\text{GP} \in H^4 (\bR^3)$ for simplicity.
\item One can use \cite{NT, BSS2} to compute $1- \langle \ph_{\text{GP}}, \gamma_N^{(1)}\ph_{\text{GP}}\rangle=O(N^{-1})$ in the ground state of $ H_N^{\text{trap}}$ explicitly, up to subleading errors of order $ o(N^{-1})$ as $N\to\infty$. This follows from arguments presented in \cite{BBCS4} (in fact, based on \cite{BBCS4}, one can derive second order expressions for reduced particle densities at low temperature in the trace class topology \cite{BLN}). In contrast to that, it remains an interesting open question whether the time evolved condensate depletion $ 1 - \langle \ph_t , \gamma^{(1)}_{N,t} \ph_t \rangle$ is similarly determined by the quasi-free evolution derived in \cite{COS}. The methods developed in the present paper may be helpful in this context and we hope to address this point in some future work.
\end{enumerate}

In Section \ref{sec:outline}, we outline the strategy of our proof and we conclude Theorem \ref{thm:main} based on a technical auxiliary result, Proposition \ref{prop:aux}, which is proved in Section \ref{sec:aux}. Standard results on the variational problem \eqref{def:a} and its minimizer, on the solution of the time-dependent Gross-Pitaevskii equation \eqref{eq:GPdyn} and on basic Fock space operators are summarized for completeness in Appendices \ref{app:dyn}, \ref{app:fock} and \ref{app:kern}. We stress that similar results as in Appendices \ref{app:dyn}, \ref{app:fock} and \ref{app:kern} have been explained in great detail in several previous and related works on the derivation of effective dynamics, see \eg \cite{BDS, BCS, BS, BNNS}.

\section{Outline of Strategy and Proof of Theorem \ref{thm:main}} \label{sec:outline}
In this section we explain the proof of Theorem \ref{thm:main}. Our approach is based on ideas previously developed in \cite{BS}, which we now briefly recall and which are most conveniently formulated using basic Fock space operators. To this end, let us start with the identity 
		\be \label{eq:linkNbot}1- \langle \ph_t, \gamma_{N,t}^{(1)}\ph_t\rangle = N^{-1 }\langle   \cN_{\bot \ph_t}  \rangle_{\psi_{N,t}}, \ee
where $\cN_{\bot \ph_t}$ denotes the number of excitations orthogonal to $\ph_t$, that is
		\be\label{def:Nbot} \cN_{\bot \ph_t} = N - a^*(\ph_t)a(\ph_t),  \ee
and where, in the rest of this paper, we abbreviate expectations of observables $\mathcal{O}$ in $ L^2_s(\bR^{3N})$ by $ \langle \mathcal O\rangle_{\phi_N} = \langle\phi_N, \mathcal O\phi_N\rangle$. In \eqref{def:Nbot}, the operators $ a^*(f), a(g) $, for $f,g \in L^2(\bR^3)$, denote the bosonic creation and annihilation operators that are defined by
		\[ \begin{split} 
		(a^* (f) \Psi)^{(n)} (x_1, \dots , x_n) &= \frac{1}{\sqrt{n}} \sum_{j=1}^n f(x_j) \Psi^{(n-1)} (x_1, \dots , x_{j-1}, x_{j+1} , \dots , x_n), \\
		(a (g) \Psi)^{(n)} (x_1, \dots , x_n) &= \sqrt{n+1} \int \overline{g} (x) \Psi^{(n+1)} (x,x_1, \dots , x_n),
		\end{split} \]
for all $ \Psi = (\Psi_0, \Psi_1,\ldots)\in \cF = \mathbb{C}\oplus \bigoplus_{n=1}^\infty L^2_s(\bR^{3n}) $, in particular for $\psi_N\in L_s^2(\bR^{3N})\hookrightarrow \cF$. Note that $ a^*(f)a(g):L_s^2(\bR^{3N})\to L_s^2(\bR^{3N})$ is bounded and preserves the number of particles $N$, for every $f,g\in L^2(\bR^3)$. Moreover, we have the commutation relations
		\[ [a(f), a^*(g)] =\langle f,g\rangle, \hspace{0.5cm} [a(f), a(g)]=0,\hspace{0.5cm}   [a^*(f), a^*(g)]  =0\]
for all $f,g\in L^2(\bR^3)$. Further results on the creation and annihilation operators and their distributional analogues $ a_x, a^*_y$, for $x,y\in\bR^3$, defined through 
		\be\label{def:aas} a(f) = \int dx\, \overline{f} (x)  a_x  , \quad a^* (g) = \int dy\, g(y)  a_y^*  ,  \ee
are collected in Appendix \ref{app:fock}. 

Based on \eqref{eq:linkNbot} and the assumption on $o_1(1)$ in \eqref{eq:ass}, a natural first attempt to prove Theorem \ref{thm:main} might consist in trying to control the growth of the number of excitations $ \cN_{\bot \ph_t}$ based on Gronwall's lemma. However, when examining the derivative 
		\[ \partial_t \big \langle  \cN_{\bot \ph_t} \big\rangle_{\psi_{N,t}}  = \big \langle   [iH_N, \cN_{\bot \ph_t}] \big\rangle_{\psi_{N,t}}  - 2\text{Re}\, \big\langle   a^*(\partial_t \ph_t)a(\ph_t) \big\rangle_{\psi_{N,t}}, \]
one soon realizes that $ [H_N, \cN_{\bot \ph_t}]$ contains several contributions of size $ O(N)$. This is actually not very surprising and a consequence of the fact that $ \psi_{N,t}$ contains short scale correlations related to \eqref{def:a}: heuristically, $ \psi_{N,t}$ can be thought of as a wave function 
		\be \label{eq:corr}  \psi_{N,t} \approx C  \prod_{1\leq i <j \leq N} f(N(x_i-x_j)) \ph_t^{\otimes N}, \ee
where $C$ is a normalization constant and $ f $ solves the zero energy scattering equation 
		\be\label{eq:0en}  -2\Delta f + Vf =0 \ee
with $\lim_{|x|\to\infty}f(x) = 1$. Notice that $f$ minimizes the functional on the \rhs in \eqref{def:a} and that $ f(N.)$ solves the zero energy scattering equation with rescaled potential $V_N$. Further properties of $f$ and related functions are summarized in Appendix \ref{app:kern}. 

Although the correlations $ f(N(x_i-x_j))$ live on a short length scale of order $O(N^{-1})$, basic computations imply that the orthogonal excitations in states as in \eqref{eq:corr} carry a large energy of size $O(N)$, prohibiting a naive control of $ \cN_{\bot \ph_t}$. On the other hand, if one could factor out these correlations, one would remain with a state closer to $\ph_t^{\otimes N}$. In this case, the number and energy of the excitations around $\ph_t$ should be easier to control. Motivated by this heuristics, the main idea of \cite{BS} is to approximate $ \psi_{N,t}$ by 
		\be 
		\begin{split}
		\psi_{N,t}    & \approx C \, \prod_{1\leq i<j\leq N} \Big(1 - (1-f)(N(x_i-x_j))\Big) \ph_t^{\otimes N}\\
		& \approx C \, \Big( 1 - \sum_{1\leq i <j \leq N}  (1-f)(N(x_i-x_j))+\ldots  \Big) \ph_t^{\otimes N}\\
		&\approx C \exp \Big( -  \frac12\int dxdy\,  (1-f)(N(x-y)) a^*_x a^*_y a_x a_y  \Big)  \ph_t^{\otimes N} \approx e^{B_t}  \ph_t^{\otimes N}.
		\end{split}
		\ee
This incorporates the expected correlation structure into the product state $ \ph_t^{\otimes N}$ by applying a unitary, generalized Bogoliubov transformation $ e^{B_t}$ with exponent
		\[ B_t = -\frac12 \int dxdy\,\Big( (1-f)(N(x-y)) \ph_t(x)\ph_t(y) a^*_x a^*_y a(\ph_t)a(\ph_t) - \text{h.c.}\Big). \]

In other words, we expect the state $e^{-B_t}\psi_{N,t}$ to behave approximately like a product state and the main result of \cite{BS} is to establish this intuition rigorously. Ignoring minor technical details, mathematically this is achieved by controlling the number and energy of excitations around $ \ph_t$ \wrt the fluctuation dynamics $ \cU_{N,t}  = e^{-B_t} e^{-iH_Nt} $ that satisfies
		\[   i\partial_t \,\cU_{N,t} =  \mathcal S_{N,t} \, \cU_{N,t} = \Big(   e^{-B_t} H_N e^{B_t} +(i\partial_t e^{-B_t})e^{B_t} \Big ) \,\cU_{N,t}.  \]
As turns out, the energy of the excitations is comparable to $\cS_{N,t}' =\cS_{N,t}- c_{N,t}$ for a suitable constant $ c_{N,t}$, so that the main result of \cite{BS} can be recast as a Gronwall bound
		\be  \label{eq:gron} \partial_t \big \langle \cS_{N,t}' + \cN_{\bot \ph_t}  \big\rangle_{\cU_{N,t}\psi_N}  \lesssim \big \langle   \cS_{N,t}' + \cN_{\bot \ph_t}  \big\rangle_{\cU_{N,t}\psi_N}. \ee
		
Although conceptually straightforward, the main difficulty of the above strategy consists in the fact that the action of $ e^{-B_t}(\cdot)e^{B_t}$, that is needed to compute $ \cS_{N,t}$, is not explicit. The novelty of \cite{BS} has therefore been to analyse $ e^{-B_t}(\cdot)e^{B_t}$ in detail, providing an explicit description of $ \cS_{N,t}$ in terms of a convergent commutator series expansion. This can be used to explicitly evaluate the commutator $ [\cS_{N,t}, \cN_{\bot \ph_t}]$ that occurs on the left hand side in \eqref{eq:gron} and this is crucial to close the Gronwall argument. 

The drawback of this method is that the series expansions are rather involved and produce a large number of irrelevant error terms. It would therefore be quite desirable to extract only the relevant terms without the need for operator exponential expansions, similarly as in \cite{Br, BBCO} in the spectral setting. Our key observation in this regard is that \eqref{eq:gron} is essentially equivalent to controlling the modified energy and excitation operators
		\be \label{eq:gronnew} 
		\begin{split}
		 &\big \langle  e^{B_t}\big( \cS_{N,t}' + \cN_{\bot \ph_t} \big) e^{-B_t}\,  \big\rangle_{\psi_{N,t}}  
		  \approx \big \langle    \cH_N \big\rangle_{\psi_{N,t} } +  \big\langle   \mathcal Q_\text{ren} \big\rangle_{\psi_{N,t} }+  \big\langle  \cN_{\text{ren}}  \big\rangle_{\psi_{N,t} },
		\end{split}
	 	\ee
where we have inserted heuristically several approximations from \cite{BS}. In \eqref{eq:gronnew}, we set 
		\be \label{def:cHN} \cH_{N} =  H_N - Ne_{\text{GP}} -  \big(a^*(\ph_t) a(Q_t i\partial_t \ph_t)+\text{h.c.}\big)  + \langle i\partial_t \ph_t,\ph_t\rangle\cN_{\bot \ph_t}  \ee
and $ e_\text{GP} \equiv \cE_{\text{GP}}(\ph_t) $ for the translation invariant energy functional $\cE_{\text{GP}}$, defined by
		\[\cE_{\text{GP}} (\ph) = \int  dx\, \Big(  \vert \nabla \ph(x) \vert^2 +   4\pi\mathfrak{a} \vert \ph(x)\vert^4 \Big),\]
Recall that $ e_{\text{GP}} $ is a conserved quantity if $t\mapsto \ph_t$ is a sufficiently regular solution of \eqref{eq:GPdyn}, in particular under the assumptions on $\ph_{\text{GP}} $ in Theorem \ref{thm:main} (see Prop. \ref{prop:GPdyn}). 

In \eqref{eq:gronnew}, we have furthermore introduced renormalized excitation operators  
		\be\begin{split}
		\label{def:Nren} 
		 \cN_{\text{ren}} &=  \cN_{\bot \ph_t}  + \int dx dy \,  \Big(  k_t (x,y)   a^*(Q_{t,x}) a^*(Q_{t,y}) \frac{a(\ph_t)}{\sqrt{N} }\frac{a(\ph_t)}{\sqrt N}  + \text{h.c.}\Big) , \\
		\mathcal Q_\text{ren} &= \frac12\int dx dy\,  i \partial_t k_t (x,y)   a^*(Q_{t,x}) a^*(Q_{t,y}) \frac{a(\ph_t)}{\sqrt{N} }\frac{a(\ph_t)}{\sqrt N} +\text{h.c.} 
		\end{split} \ee	
in terms of orthogonal excitation fields $ a^*(Q_{t,x}),  a(Q_{t,y})$, defined as follows: denoting by $Q_t$ the projection $ Q_t = 1-|\ph_t \rangle\langle\ph_t |$ onto the orthogonal complement of $\ph_t$, we set
		\be \label{def:aQ} a^*(Q_t f) = \int dx\, f(x) a^*(Q_{t,x}), \hspace{0.5cm}  a(Q_t g) = \int dx\, \overline g (x) a^*(Q_{t,x}). \ee
It is then straightforward to verify that $\cN_{\bot \ph_t} = \int dx \,a^*(Q_{t,x})a(Q_{t,x})$ and that
		\[\begin{split}  a^*(Q_{t,x}) &= a^*_x - \overline \ph_t(x) a^*(\ph_t), \;\;\;\;  a(Q_{t,y}) = a_y -  \ph_t(y) a(\ph_t), \\
		  [a(Q_{t,x}), a^*(Q_{t,y}) ]  &= [a_x, a^*(Q_{t,y}) ]=  Q_t(x,y), \;\;   [a(Q_{t,x}), a^*(\ph_t) ]  = (Q_t \ph)(x) = 0, \end{split} \]
where $  Q_t(x,y) = \delta(x,y)- \ph_t(x)\overline{\ph_t(y)}$ denotes the integral kernel of $Q_t$. 

Finally, fixing some $\chi \in C_c^\infty (B_{2r}(0))$ with $ \chi_{|  B_{r}(0)}\equiv 1$, we define the kernel $k_t$ by
		\be \begin{split} \label{def:k} &(x,y) \mapsto k_t(x,y)   = N(1-f)(N(x-y))\chi(x-y) \ph_t(x)\ph_t(y) \in L^2(\bR^3\times \bR^3).
		\end{split} \ee

Notice that $\cH_N, \cN_{\text{ren}}$ and $ \cQ_\text{ren}$ are time-dependent. For simplicity, we suppress this dependence in our notation. Moreover, we remark that the cutoff $\chi$ in the definition of $k_t$ is for technical reasons only (we ignored this technicality in the heuristic arguments outlined above). Basic properties of the kernel $ k_t$ are collected in Appendix \ref{app:kern}. 

We assume throughout the remainder that the radius $r>0$, related to $\chi \in C_c^\infty (B_{2r}(0))$ in \eqref{def:k}, is chosen sufficiently small, but fixed (independently of $N$). As explained below in Lemma \ref{lm:aux}, this implies that\footnote{By $ \pm A \leq B$ for self-adjoint operators $A,B$, we abbreviate that $-B\leq A\leq B$. Moreover, generic constants independent of $N$ and $t$ are typically denoted by $c,C>0$ and may vary from line to line.} for some $C>0$ and every $t\in\bR$ it holds true that
		\be\label{eq:NNren} C^{-1} (\cN_{\bot\ph_t} +1) \leq   ( \cN_{\text{ren}} +1 ) \leq C  ( \cN_{\bot\ph_t} +1), \hspace{0.5cm} \pm \cQ_\text{ren}\leq C e^{C|t|} ( \cN_{\text{ren}} +1 ). \ee

Having introduced all objects that are relevant in the sequel, let us briefly comment on the heuristics underlying the approximation \eqref{eq:gronnew}. What \cite{BS} has shown rigorously is that transformations $ e^{B_t} $ as above act on creation and annihilation operators approximately like standard Bogoliubov transformations. It then turns out that $ e^{-B_t} (\cdot) e^{B_t}$ regularizes certain singular contributions to $ H_N $, and these renormalizations are essentially obtained from the contributions linear in $B_t$ when expanding $ e^{-B_t} a_x e^{B_t} \approx a_x  + [ a_x, B_t ] $. In \eqref{eq:gronnew}, we simply inserted this linear approximation on the level of $ L^2_s(\bR^{3N})$.   
	
Finally, let us point out that it is straightforward to compute the time derivative of the right hand side in \eqref{eq:gronnew} explicitly - in strong contrast to the computation of the left hand side in \eqref{eq:gron}. This naturally raises the question whether a Gronwall bound can be proved directly on the right hand side of \eqref{eq:gronnew}, avoiding the use of operator exponential expansions altogether, similarly as in \cite{BFS, Br, BBCO} in the spectral setting. On the technical level, this is our main contribution and it leads to the following result. 

\begin{prop}\label{prop:aux} 
Let $ \cH_N$ be as in \eqref{def:cHN} and set $  \psi_{N,t}= e^{-iH_Nt}\psi_N$ for $t\in\bR$ and initial data $\psi_N\in L^2_s(\bR^{3N})$ as in Theorem \ref{thm:main}. Then, for suitable $c, C>0$, we have that 
		\[ \cN_{\bot\ph_t} \leq  \cH_N + \cQ_\emph{ren} +  Ce^{C|t|}  ( \cN_{\emph{ren}} +1)\]
as well as the Gronwall bound
		\[ \partial_t \,  \big \langle    \cH_N +\cQ_\emph{ren} +  Ce^{C|t|}   (\cN_{\emph{ren}} +1)  \big\rangle_{\psi_{N,t} }  \leq c \,e^{c |t|} \, \big \langle    \cH_N + \cQ_\emph{ren}  +  Ce^{C|t|}  (  \cN_{\emph{ren}} +1) \big\rangle_{\psi_{N,t} }.\]
\end{prop}
Assuming the validity of Proposition \ref{prop:aux}, whose proof is explained in detail in the next Section \ref{sec:aux}, we conclude this section with the proof of Theorem \ref{thm:main}. 

\begin{proof}[Proof of Theorem \ref{thm:main}] 
This was already explained in \cite{BS}; we recall the main steps. Without loss of generality assume $t\geq 0$. By \eqref{eq:linkNbot}, note that \eqref{eq:main} is equivalent to 
		\be\label{eq:proofgoal}  \big\langle  \cN_{\bot\ph_t}  \big\rangle_{\psi_{N,t} } \leq C \big( 1+N o_1(1) + N o_2(1) \big) \exp(c \exp c\,t). \ee
By Prop. \ref{prop:aux}, Gronwall's lemma and the bound \eqref{eq:NNren}, we know that 
		\[\begin{split}
		 \big\langle  \cN_{\bot\ph_t}  \big\rangle_{\psi_{N,t} } \leq  \big \langle  \cH_N +\cQ_\text{ren} +  Ce^{Ct}  ( \cN_{\text{ren}}+1) \big\rangle_{\psi_{N,t} } & \leq c_t \,\big \langle  (\cH_{N} +\cQ_\text{ren} +C\cN_{\text{ren}}+1)_{|t=0} \big\rangle_{\psi_{N} }\\
		 & =  c_t \, \big \langle  (\cH_{N}+\cQ_\text{ren})_{|t=0}  + \cN_{\bot \ph_{\text{GP} } }  + 1 \big\rangle_{\psi_{N} }   
		\end{split}\]
for some time-dependent constant $c_t \leq  C \exp(c\,\exp(c\,t))$. Hence, it is enough to analyze $  \langle  (\cH_N)_{|t=0}\rangle_{\psi_{N} } $, $  \langle  (\cQ_\text{ren})_{|t=0}\rangle_{\psi_{N} }$  and  $ \langle \cN_{\bot \ph_{\text{GP} } } \rangle_{\psi_{N} } $. Using once again \eqref{eq:linkNbot}, we have that  
		\[ \big\langle \cN_{\bot \ph_{\text{GP} } } \big\rangle_{\psi_{N} }  = N\big (  1 - \langle \ph_\text{GP} , \gamma^{(1)}_N \ph_\text{GP} \rangle \big) = N o_2(1) \] 
and, by \eqref{eq:NNren}, that
		\[\langle  (\cQ_\text{ren})_{|t=0}\rangle_{\psi_{N} }\leq C   \big\langle \cN_{\bot \ph_{\text{GP} }  } +1\big\rangle_{\psi_{N} }  = C\big (1+ N o_2(1) \big). \] 
By \eqref{def:cHN}, on the other hand, we have that 
		\[\begin{split}
		\big\langle  (\cH_N)_{|t=0} \big\rangle_{\psi_{N} }  
		&\!\!\!\! =   \big\langle H_N \!-\! Ne_\text{GP}   - \! \big(a^*(\ph_t) a(Q_t i\partial_t \ph_t)+\text{h.c.}\big)_{|t=0} \!+\! \langle i\partial_t \ph_t,\ph_t\rangle_{|t=0} \,\cN_{\bot \ph_\text{GP}}\big\rangle_{\psi_{N}}\\
		&\!\!\!\! =  \big\langle H_N \!-\!  Ne_\text{GP}   - \! \big(a^*(\ph_t) a(Q_t i\partial_t \ph_t)+\text{h.c.}\big)_{|t=0}\big\rangle_{\psi_{N}}  \!+\!  \langle i\partial_t \ph_t,\ph_t\rangle_{|t=0} \,N o_2(1)
		\end{split}\]
and, by \eqref{eq:GPdyn}, that $\langle i\partial_t \ph_t,\ph_t\rangle_{|t=0} = e_\text{GP}+ 4\pi\mathfrak{a}\|\ph_{\text{GP}}\|_4^4    =O(1) $. Since we assume that $\ph_\text{GP}$ minimizes the Gross-Pitaevskii functional $\cE_\text{GP}^\text{trap}$, it solves the Euler-Lagrange equation  
		\[  \big( -\Delta + V_{\text{ext}} + 8\pi\mathfrak{a}|\ph_\text{GP}|^2\big) \ph_\text{GP} =   \mu_\text{GP} \ph_\text{GP}, \hspace{0.5cm}  \mu_\text{GP}= e_\text{GP}+ 4\pi\mathfrak{a}\|\ph_{\text{GP}}\|_4^4.  \]
Combining this with \eqref{eq:GPdyn}, we then find 
		\[\begin{split}
		& -  \big\langle \big(a^*(\ph_t) a(Q_t i\partial_t \ph_t)+ a^*(Q_t i\partial_t \ph_t)a(\ph_t)\big)_{|t=0}\big\rangle_{\psi_{N}} \\
		&\;=  2\text{Re} \,   \big\langle a^*(\ph_\text{GP}) a(V_{\text{ext}}\ph_\text{GP}) \big\rangle_{\psi_{N}} - 2 \langle \ph_\text{GP}, V_\text{ext}\ph_\text{GP} \rangle \big \langle a^*(\ph_\text{GP}) a^*(\ph_\text{GP}) \big\rangle_{\psi_{N}} \\
		&\; = 2 N\text{Re} \langle \ph_\text{GP},  \gamma_N^{(1)} V_{\text{ext}}\ph_\text{GP}\rangle - 2 N \langle \ph_\text{GP}, V_\text{ext} \ph_\text{GP} \rangle + \langle \ph_\text{GP}, V_\text{ext}\ph_\text{GP}\rangle N o_2(1),  
		\end{split}\]
where $  \langle \ph_\text{GP}, V_\text{ext}\ph_\text{GP}\rangle \leq e_\text{GP}=O(1)$. Now, if we replace $ V_{\text{ext}}$ by $ V'_{\text{ext}} = V_{\text{ext}}+\Lambda$ for some sufficiently large $\Lambda>0$ so that $V'_{\text{ext}}\geq0$, by the assumption that $V_{\text{ext}}\in L^\infty_\text{loc}(\bR^3)$ with $ \lim_{|x|\to\infty}V_{\text{ext}}=\infty$, and use that $ 0\leq (\gamma_N^{(1)}\big)^2\leq \gamma_N^{(1)}\leq 1$, Cauchy-Schwarz implies  
		\[\begin{split}
		 2\text{Re} \,\langle \ph_\text{GP},  \gamma_N^{(1)} V_{\text{ext}}\ph_\text{GP}\rangle &\leq  \langle \ph_\text{GP}, \gamma_N^{(1)} V'_{\text{ext}}\gamma_N^{(1)}\ph_\text{GP}\rangle     + \langle \ph_\text{GP}, V'_\text{ext} \ph_\text{GP} \rangle- 2\Lambda\langle \ph_\text{GP},  \gamma_N^{(1)} \ph_\text{GP}\rangle\\
		 &\leq \langle \ph_\text{GP}, \gamma_N^{(1)} V'_{\text{ext}}\gamma_N^{(1)}\ph_\text{GP}\rangle -\Lambda + \langle \ph_\text{GP}, V_\text{ext} \ph_\text{GP} \rangle +2 \Lambda \, o_2(1)\\
		 & \leq \text{tr}\,\big|V'_{\text{ext}}\big|^{1/2} \big(\gamma_N^{(1)}\big)^2 \big|V'_{\text{ext}}\big|^{1/2} -\Lambda+ \langle \ph_\text{GP}, V_\text{ext} \ph_\text{GP} \rangle +2  \Lambda \, o_2(1)\\
		 & \leq \text{tr}\, \gamma_N^{(1)} V_{\text{ext}} + \langle \ph_\text{GP}, V_\text{ext} \ph_\text{GP} \rangle +2  \Lambda \, o_2(1). 
		 \end{split}\]
This shows that
		\[ \big\langle  (\cH_N)_{|t=0} \big\rangle_{\psi_{N} }\leq  \big\langle H_N^\text{trap}\big\rangle_{\psi_N} -  Ne_\text{GP}^\text{trap}   + CN o_2(1) \leq C \big( N o_1(1)+N o_2(1) \big). \] 
Collecting the previous bounds, we obtain \eqref{eq:proofgoal} and thus \eqref{eq:main}.
\end{proof}
\section{Renormalized Hamiltonian and Proof of Prop.\  \ref{prop:aux}}\label{sec:aux}

The goal of this section is to prove Proposition \ref{prop:aux}. Our proof is based on several lemmas that collect important properties of the operators $ \cH_N, \cN_\text{ren}$ and $ \cQ_\text{ren}$, defined in \eqref{def:cHN} and \eqref{def:Nren}, respectively. We start with the proof of the bound \eqref{eq:NNren} and the derivation of the leading order contributions to $ \partial_t \, \cN_\text{ren}$ and $ \partial_t  \cQ_\text{ren}$. 

\begin{lemma}\label{lm:aux}
Let $\cN_{\emph{ren}}, \cQ_{\emph{ren}} $ be as in \eqref{def:Nren} and choose $\chi \in C_c^\infty (B_{2r}(0))$, $ \chi_{|  B_{r}(0)}\equiv 1$ in \eqref{def:k} so that $r>0$ is small enough. Then, for some $C>0$ and every $t\in\bR$, we have
		\be\label{eq:lmaux} C^{-1} (\cN_{\bot\ph_t} +1) \leq  ( \cN_{\emph{ren}}+1)  \leq C  ( \cN_{\bot\ph_t} +1),\hspace{0.5cm} \pm \cQ_\emph{ren}\leq C e^{C|t|} ( \cN_{\emph{ren}} +1 ) \ee
and
		\be\begin{split}
		 \label{eq:lmaux2}  \pm \big(  \partial_t \,\cN_{\emph{ren}} - \big[  -  \big(a^*(\ph_t) a(Q_t \partial_t \ph_t)-\emph{h.c.}\big)  + \langle \partial_t \ph_t,\ph_t\rangle\cN_{\bot \ph_t}, \cN_{\emph{ren}} \big]\big)&\leq  Ce^{C|t|} (\cN_{\emph{ren}} + 1), \\
		 \pm \big(  \partial_t \cQ_{\emph{ren}} - \big[  -  \big(a^*(\ph_t) a(Q_t \partial_t \ph_t)-\emph{h.c.}\big)  + \langle \partial_t \ph_t,\ph_t\rangle\cN_{\bot \ph_t}, \cQ_{\emph{ren}} \big]\big)&\leq  Ce^{C|t|} (\cN_{\emph{ren}} + 1).    
		 \end{split} \ee
\end{lemma}
\begin{proof} 
We recall that 
		\[\begin{split}
		\cN_{\text{ren}} & = \cN_{\bot \ph_t}  + \int dx dy \,  \Big(  k_t (x,y)   a^*(Q_{t,x}) a^*(Q_{t,y}) \frac{a(\ph_t)}{\sqrt{N} }\frac{a(\ph_t)}{\sqrt N}  + \text{h.c.}\Big).
		\end{split}\]
By Lemma \ref{lm:ker}, we have that $\sup_{t\in\bR}  \| k_{t}\|\leq C r^{1/2} $. If we combine this with the trivial bound  $0\leq  a^*(\ph_t)a(\ph_t)\leq N$ and the operator bounds of Lemma \ref{lm:fock}, we obtain  
		\[  \cN_{\bot \ph_t}  (1 - C r^{1/2}) - Cr^{1/2}   \leq   \cN_{\text{ren}} \leq  \cN_{\bot \ph_t}  (1+ C r^{1/2})    + Cr^{1/2}  \]
for some $C>0$ independent of $r>0$ and $t\in \bR$. The bound for $\cQ_\text{ren}$ follows similarly.  

To prove \eqref{eq:lmaux2}, we first analyze $ \partial_t \,\cN_{\text{ren}}$, based on the above decomposition of $\cN_\text{ren}$. Using \eqref{lm:aux} and the bounds in Lemmas \ref{lm:fock} and \ref{lm:ker}, observe that all operators occurring in $ \partial_t \,\cN_{\text{ren}}$ that only contain the fields $ a^\sharp(Q_{t,x})$ or normalized factors $ a^\sharp(\ph_t)/ \sqrt{N}$, $a^\sharp(\partial_t \ph_t)/ \sqrt{N}$ can be bounded by $   Ce^{C|t|} (\cN_{\text {ren}} + 1)$. The remaining contributions must contain at least one factor $ a^\sharp (\ph_t)$ (without the $1/{\sqrt N}$ normalization). Using that
		\be\label{eq:dtQ} 
		\begin{split} \partial_t Q_t  &=  - \big( | \ph_t \rangle \langle Q_t  \partial_t \ph_t | +\text{h.c.}\big) + 2\text{Re}\langle\partial_t \ph_t,\ph_t\rangle | \ph_t \rangle \langle  \ph_t | = - \big( | \ph_t \rangle \langle Q_t  \partial_t \ph_t | +\text{h.c.}\big), 
		\end{split} \ee
we thus find 
		\[\begin{split}
		\partial_t \, \cN_{\text{ren}}  & = - 2  \int dx dy \,  \Big(  \overline{ Q_t \partial_t \ph_t}(x) k_t (x,y)   a^*(\ph_t) a^*(Q_{t,y}) \frac{a(\ph_t)}{\sqrt{N} }\frac{a(\ph_t)}{\sqrt N}  + \text{h.c.}\Big)  \\
		&\hspace{0.5cm}-\big( a^*(\ph_t) a(\partial_t \ph_t) +\text{h.c.}\big) + \cE_1, 
		\end{split}\]
up to an error $\cE_1$ bounded by $ \pm \cE_1 \leq Ce^{C|t|} (\cN_{\text{ren}} + 1)$. We proceed in the same way to extract the main contributions to the commutator on the \lhs in \eqref{eq:lmaux2}, using that
		\[\begin{split}
		[\cN_{\bot\ph_t}, a^*(Q_{t,x}) a(\ph_t) ] &=  a^*(Q_{t,x}) a(\ph_t), \;\; [\cN_{\bot\ph_t}, a^*(\ph_t) a(\ph_t)  ]  = 0  
		\end{split}\]
Then, the same argument as above yields
		\[\begin{split}
		& \big[  -  \big(a^*(\ph_t) a(Q_t \partial_t \ph_t)-\text{h.c.}\big)  + \langle \partial_t \ph_t,\ph_t\rangle\cN_{\bot \ph_t}, \cN_{\text{ren}}\big]\\
		& =  - 2  \int dx dy \,  \Big(  \overline{ Q_t\partial_t \ph_t}(y) k_t (x,y)   a^*(\ph_t) a^*(Q_{t,y}) \frac{a(\ph_t)}{\sqrt{N} }\frac{a(\ph_t)}{\sqrt N}  + \text{h.c.}\Big)  \\
		&\hspace{0.5cm}-\big( a^*(\ph_t) a(Q_t\partial_t \ph_t) +\text{h.c.}\big) + \cE_2
		\end{split}\]
up to an error $ \pm \cE_2 \leq Ce^{C|t|} (\cN_{\text{ren}} + 1)$. Comparing this with $\partial_t \, \cN_{\text{ren}}$ and using that 
		\[ a^*(\ph_t) a(\partial_t \ph_t) + a^*(\partial_t \ph_t) a(\ph_t) =  a^*(\ph_t) a( Q_t \partial_t \ph_t) +a^*(Q_t\partial_t \ph_t) a(\ph_t), \] 
which follows from $ \text{Re}\langle \partial_t\ph_t, \ph_t\rangle =0 $ by mass conservation, this proves the first bound in \eqref{eq:lmaux2}. For the analogous bound on $\cQ_\text{ren}$, we proceed in the same way and find 
		\[\begin{split}
		\partial_t \cQ_\text{ren}&\!   = -  \int dx dy \,  \Big(  \overline{ Q_t\partial_t \ph_t}(x) i\partial_t k_t (x,y)   a^*(\ph_t) a^*(Q_{t,y}) \frac{a(\ph_t)}{\sqrt{N} }\frac{a(\ph_t)}{\sqrt N}  + \text{h.c.}\Big) +\cE_3\\
		& =  \big[  -  \big(a^*(\ph_t) a(Q_t \partial_t \ph_t)-\text{h.c.}\big)  + \langle \partial_t \ph_t,\ph_t\rangle\cN_{\bot \ph_t}, \cQ_{\text{ren}} \big] + \cE_4,
		\end{split}\]
up to errors $ \cE_3, \cE_4$ bounded by $ \pm \cE_3\leq Ce^{C|t|} (\cN_{\text{ren}} + 1)$, $\pm \cE_4\leq Ce^{C|t|} (\cN_{\text{ren}} + 1)$.
\end{proof}

The next lemma is the first of two key ingredients in the proof of Proposition \ref{prop:aux}. It compares the operator $\cH_{N}$, defined in \eqref{def:cHN}, to a renormalized Hamiltonian $ \cH_{\text{ren}}$ which equals the sum of the kinetic and potential energies of orthogonal excitations relative to renormalized field operators $ b_{x}, b^*_{y}$, which are defined by 
		\be\begin{split} \label{def:newflds} 
		b_{x} &=  a(Q_{t,x}) +  \int  dz\,  (Q_t\otimes Q_t k_t) (x,z) \,a^*_z \,\frac{a(\ph_t)}{\sqrt{N} }\frac{a(\ph_t)}{\sqrt N}, \\
		 b^*_{y} &=  a^*(Q_{t,y}) +  \frac{a^*(\ph_t)}{\sqrt{N} }\frac{a^*(\ph_t)}{\sqrt N} \int  dz\, \overline{(Q_t\otimes Q_t k_t)} (y ,z) \,a_z. 
		\end{split} \ee
Note that this is analogous to \cite[Eq. (11)]{Br}. In terms of these new fields, we set
		\be \label{def:ren1}  \begin{split}
		\cK_{\text{ren}} & = \int dx\, b^*_{x} (-\Delta_x) b_{x}, \;\;
		\cV_{\text{ren}}  =  \frac12 \int dx dy\, N^2V(N(x-y))b^*_{x}a^*(Q_{t,y}) a(Q_{t,y})b_{x}
		\end{split}\ee
as well as $ \cH_{\text{ren}} = \cK_{\text{ren}}  + \cV_{\text{ren}}$. Note that $\cH_{\text{ren}}\geq 0$ since both $ \cK_{\text{ren}}\geq 0$ and $ \cV_{\text{ren}}\geq0$. Note, moreover, that $ \cN_\text{ren}$ equals $ \int dx\, b^*_x b_x$, up to a correction which is quadratic in $k_t$. 

\begin{lemma} \label{lm:dec}
The operator $\cH_{N}$, defined in \eqref{def:cHN}, satisfies  
		\be \label{eq:cHNbnd} \frac12 \cH_{\emph{ren}} - Ce^{C|t|} (\cN_{\emph{ren}} +1)  \leq  \cH_{N} \leq 2\, \cH_{\emph{ren}} + Ce^{C|t|} (\cN_{\emph{ren}} +1). \ee
Moreover, we have that 
		\be \label{eq:cmcHNbnd} \begin{split}
		\pm  \big[ i \cH_N, \cN_{\emph{ren}}\big] &\leq  C\,\cH_{\emph{ren}}  + Ce^{C|t|} (\cN_{\emph{ren}} +1).  
		\end{split} \ee
\end{lemma}
\begin{proof}
We begin with the operator bounds in \eqref{eq:cHNbnd}. The proof consists essentially of two main steps. First, we split $\cH_{N}$ into several parts according to condensate and orthogonal excitation contributions to the energy. In terms of the $ a_x, a^*_y$, the Hamiltonian $H_N$ reads
		\[H_N = \int dx\, a_x^* (-\Delta_x) a_x + \frac12 \int dxdy\, N^2V(N(x-y)) a^*_x a^*_y a_x a_y.  \]
We split $ a_x = a(Q_{t,x}) + \ph_t(x) a(\ph_t)$, $ a^*_y = a^*(Q_{t,y}) + \overline\ph_t(y) a^*(\ph_t)$, insert this into $H_N$ and then expand $\cH_{N}$ into the sum $ \cH_{N} = \sum_{j=0}^4 \cH_{N}^{(j)}$, where 
		\be\begin{split} \label{eq:HNdec}
		\cH_{N}^{(0)}&= \frac N2  \big\langle \ph_t, (N^3V(N.) \ast|\ph_t|^2) \ph_t \big\rangle  \frac{a^*(\ph_t)}{\sqrt{N}}\frac{a^*(\ph_t)}{\sqrt{N}} \frac{a(\ph_t)}{\sqrt{N}}\frac{a(\ph_t)}{\sqrt{N}} -Ne_{\text{GP}} \\
		&\hspace{0.4cm} +  N \langle \ph_t, -\Delta \ph_t\rangle \frac{a^*(\ph_t)}{\sqrt{N}} \frac{a(\ph_t)}{\sqrt{N}} +  \langle i\partial_t \ph_t,\ph_t\rangle\cN_{\bot \ph_t}  \\
		\cH_{N}^{(1)}&=  a^*(\ph_t) a \big(Q_t ( N^3V (N.) \ast|\ph_t|^2) \ph_t) \big)-   a^*(\ph_t) a \big(Q_t  ( 8\pi \mathfrak{a} |\ph_t|^2 \ph_t) \big)\\
		&\hspace{0.5cm} -  \frac{a^*(\ph_t)}{\sqrt{N}} a \big(Q_t ( N^3V(N.) \ast|\ph_t|^2) \ph_t) \big) \frac{\cN_{\bot \ph_t} }{\sqrt{N}} +\text{h.c.}, \\
		\cH_{N}^{(2)} &= \int dx\, a^*(Q_{t,x}) (-\Delta_x) a(Q_{t,x}) \\
				  &\hspace{0.3cm} +\int dx dy  \, N^3 V(N(x-y)) |\ph_t (y)|^2   a^*(Q_{t,x}) a (Q_{t,x}) \frac{a^*(\ph_t)}{\sqrt{N}}\frac{a(\ph_t)}{\sqrt{N}} \\
				  &\hspace{0.3cm}+\int dx dy  \, N^3 V(N(x-y)) \ph_t (x) \overline{\ph}_t (y) a^*(Q_{t,x}) a(Q_{t,y})  \frac{a^*(\ph_t)}{\sqrt{N}}\frac{a(\ph_t)}{\sqrt{N}} \\
    				  &\hspace{0.3cm} +\!   \frac12\int \!dx dy\, N^3V(N(x-y)) \!\Big( \ph_t (x)\ph_t(y) a^*(Q_{t,x}) a^*(Q_{t,y}) \frac{a(\ph_t)}{\sqrt{N}}\frac{a(\ph_t)}{\sqrt{N}}\!+\!\text{h.c.} \!\Big), \\
		\cH_{N}^{(3)} &= \int dx dy \, N^{5/2} V(N(x-y)) \ph_t (y) a^*(Q_{t,x}) a^*(Q_{t,y})  a(Q_{t,x}) \frac{a(\ph_t)}{\sqrt{N}} + \text{h.c.} , \\
		\cH_{N}^{(4)} &= \frac1{2}\int dx dy\, N^2V(N(x-y)) a^*(Q_{t,x}) a^*(Q_{t,y}) a(Q_{t,y}) a(Q_{t,x}) . \\
		\end{split} \ee
Here, we normalized the $a^\sharp (\ph_t)$ by a factor $\sqrt{N}$, anticipating that $ \langle a^*(\ph_t)a(\ph_t)\rangle_{\psi_{N,t}}\approx N$. 
				
In the second step, we extract $ \cK_{\text{ren}}$ and $\cV_{\text{ren}}$ from $\cH_{N}$, up to errors controlled by $ \cH_{\text{ren}}$ and $\cN_{\text{ren}}$. The error estimates are mostly straightforward applications of Cauchy-Schwarz in combination with the results of Appendices \ref{app:fock} and \ref{app:kern}. Below, we outline the key steps since most of the bounds have already been explained at length in \eg \cite{BDS, BCS, BS}. 

Now, as shown below, the main contributions to $ \cH_{N}^{(0)} $ and $ \cH_{N}^{(1)}$ are cancelled so let us switch directly to $\cH_{N}^{(2)} $ which contains $\cK_{\text{ren}}$. Abbreviating in the following 
		$$ j_x(\cdot) = j(x, \cdot) =  j(\cdot, x)  $$ 
for symmetric $ j\in L^2(\bR^3\times \bR^3)$, we rewrite 
		\[\begin{split}
		&\int dx\, a^*(Q_{t,x}) (-\Delta_x) a(Q_{t,x}) \\
		& =  \int dx\, \Big(b^*_{x} - \frac{a^*(\ph_t)}{\sqrt{N} }\frac{a^*(\ph_t)}{\sqrt N} a\big(( Q_t\otimes Q_t k_t)_{x}) \big) \Big) \\
		&\hspace{1.2cm} \times(-\Delta_x) \Big(b_{x} -  a^*\big(( Q_t\otimes Q_t k_t)_{x}) \big)  \frac{a(\ph_t)}{\sqrt{N} }\frac{a(\ph_t)}{\sqrt N}\Big)\\
		& = \cK_{\text{ren}} + \int dx dy\, \Big( (\Delta_x( Q_t\otimes Q_t k_t)_{x})(y) \,b^*_{x}  a^*(Q_{t,y})  \frac{a(\ph_t)}{\sqrt{N} }\frac{a(\ph_t)}{\sqrt N} +\text{h.c.}\Big)\\
		&\hspace{0.5cm} + \int dx\,  \big\langle ( Q_t\otimes Q_t k_t)_{x} , -\Delta_x ( Q_t\otimes Q_t k_t)_{x} ) \big\rangle  \frac{a^*(\ph_t)}{\sqrt{N}}\frac{a^*(\ph_t)}{\sqrt{N}} \frac{a(\ph_t)}{\sqrt{N}}\frac{a(\ph_t)}{\sqrt{N}} + \cE_1   
		\end{split} \]
for an error $\cE_1\geq 0$ which is bounded by 
		\[\begin{split}
		 \cE_1 &= \int dx\,  \frac{a^*(\ph_t)}{\sqrt{N} }\frac{a^*(\ph_t)}{\sqrt N} a^*(\nabla_x( Q_t\otimes Q_t k_t)_{x})    a(\nabla_x ( Q_t\otimes Q_t k_t)_{x})  \frac{a(\ph_t)}{\sqrt{N} }\frac{a(\ph_t)}{\sqrt N} \\
		 &\leq   \int dx\, a^*(\nabla_x( Q_t\otimes Q_t k_t)_{x}) a(\nabla_x( Q_t\otimes Q_t k_t)_{x}) = \!\int dx dy\,  g_t(x,y) a^*(Q_{t,x}) a(Q_{t,y})\\
		 & \leq Ce^{C|t|} (\cN_{\text{ren}} +1).
		\end{split} \]
Here, we used Lemma \ref{lm:aux}, Lemma \ref{lm:fock} and Lemma \ref{lm:ker}, implying $\|g_t\|\leq C$ for  
		\[\begin{split}
		g_t(x,y) &= \int dz\, (\nabla_2 Q_t\otimes Q_t k_t) (x,z) (\nabla_2 \overline{Q_t\otimes Q_t k_t}) (y, z).
		\end{split}\]
By Lemma \ref{lm:ker}, we also find
		\[\begin{split} \int dxdy\, \Big|& (\Delta_1k_t)(x,y) + \frac12 N^3(Vf)(N(x-y)) \ph_t(x)\ph(y) \\
		& +2 N^2\big(\nabla (1-f)\big)(N(x-y))\cdot \nabla \ph_t(x)\ph_t(y)\chi(x-y) \Big|^2\leq  C,  \end{split}\]
and this can be used to show that
		\[\begin{split}
		&\int dx\, \Big( (\Delta_x( Q_t\otimes Q_t k_t)_{x})(y) \,b^*_{x}  a^*(Q_{t,y})   \frac{a(\ph_t)}{\sqrt{N} }\frac{a(\ph_t)}{\sqrt N} +\text{h.c.}\Big)\\
		& = -\frac12 \int dx\, \Big(  N^3(Vf)(N(x-y)) \ph_t(x)\ph(y) b^*_{x}  a^*(Q_{t,y} )  \frac{a(\ph_t)}{\sqrt{N} }\frac{a(\ph_t)}{\sqrt N} +\text{h.c.}\Big) + \cE_{2}
		\end{split}\]
for an error $\cE_{2}$ which, for every $\delta>0$ and some $C>0$, is controlled by 
		\[\pm \cE_{2} \leq  \delta \cK_{\text{ren}} +  C \delta^{-1} e^{C|t|} (\cN_{\text{ren}} +1).\]
Here we used that
		\[\begin{split} 
		& \pm \Big( \int dx dy\, N^2(\nabla (1-f)\big)(N(x-y))\cdot \nabla \ph_t(x)\ph_t(y)\chi(x-y)  \,b^*_{x}  a^*(Q_{t,y} )  \frac{a(\ph_t)}{\sqrt{N} }\frac{a(\ph_t)}{\sqrt N}  \\
		&    \hspace{0.7cm} + \int dx dy\, N (1-f)(N(x-y)) \nabla \ph_t(x)\ph_t(y)\chi(x-y)  \cdot \nabla_x b^*_{x}  a^*(Q_{t,y} )  \frac{a(\ph_t)}{\sqrt{N} }\frac{a(\ph_t)}{\sqrt N} \Big)\\
		& \leq C e^{C|t|} (\cN_{\text{ren}} +1)
		\end{split}\]
by integration by parts, Cauchy-Schwarz and Lemma \ref{lm:ker}, and that 
		\[\begin{split}
		&\Big|  \int dx dy\, N (1-f)(N(x-y)) \nabla \ph_t(x)\ph_t(y)\chi(x-y)  \cdot \langle \phi_N, \nabla_x b^*_{x}  a^*(Q_{t,y} )  \frac{a(\ph_t)}{\sqrt{N} }\frac{a(\ph_t)}{\sqrt N} \phi_N\rangle \Big|\\
		&\leq C \langle \cK_{\text{ren}}\rangle_{\phi_N}^{1/2} \langle \, \cN_{\text{ren}} +1 \rangle_{\phi_N}^{1/2}
		\end{split} \]
for every $\phi_N\in L^2_s(\bR^{3N})$. Finally, Lemma \ref{lm:ker} and $ a^*(\ph_t)a(\ph_t)= N-\cN_{\bot\ph_t}$ imply that
		\[\begin{split}
		& \int dx\,   \big\langle ( Q_t\otimes Q_t k_t)_{x} , -\Delta_x ( Q_t\otimes Q_t k_t)_{x} ) \big\rangle  \frac{a^*(\ph_t)}{\sqrt{N}}\frac{a^*(\ph_t)}{\sqrt{N}} \frac{a(\ph_t)}{\sqrt{N}}\frac{a(\ph_t)}{\sqrt{N}}\\ 
		& = \frac N2  \int dx dy \,  N^3\big(Vf (1-f)\big) (N(x-y)) |\ph_t(x)|^2 |\ph_t(y)|^2 + \cE_{3},
		\end{split} \]
where $ \pm \cE_{3}  \leq Ce^{C|t|} $. Combining all this with the simple estimates
		\[\begin{split}
		&\pm  \int dx dy  \, N^3 V(N(x-y)) |\ph_t (y)|^2   a^*(Q_{t,x}) a (Q_{t,x}) \frac{a^*(\ph_t)}{\sqrt{N}}\frac{a(\ph_t)}{\sqrt{N}} \leq Ce^{C|t|} (\cN_{\text{ren}} +1),\\
		&\pm\int dx dy  \, N^3 V(N(x-y)) \ph_t (x) \overline{\ph}_t (y) a^*(Q_{t,x}) a(Q_{t,y})  \frac{a^*(\ph_t)}{\sqrt{N}}\frac{a(\ph_t)}{\sqrt{N}}  \Big)\leq Ce^{C|t|} (\cN_{\text{ren}} +1),
		\end{split}\]
which follow from Cauchy-Schwarz and Lemma \ref{lm:aux}, and the fact that 
		\[\begin{split}
		  &\frac12\int \!dx dy\, N^3V(N(x-y)) \Big( \ph_t (x)\ph_t(y) a^*(Q_{t,x}) a^*(Q_{t,y}) \frac{a(\ph_t)}{\sqrt{N}}\frac{a(\ph_t)}{\sqrt{N}}\!+\!\text{h.c.} \!\Big) \\
		  & = \frac12\int \!dx dy\, N^3V(N(x-y)) \Big( \ph_t (x)\ph_t(y) b^*_{x} a^*(Q_{t,y}) \frac{a(\ph_t)}{\sqrt{N}}\frac{a(\ph_t)}{\sqrt{N}}\!+\!\text{h.c.} \!\Big)\\ 
		  &\hspace{0.5cm} - N  \int dx dy \,  N^3 V(1-f) (N(x-y)) |\ph_t(x)|^2 |\ph_t(y)|^2 \,  + \cE_{4}
		 \end{split}\]
for an error $ \pm \cE_{4} \leq Ce^{C|t|} (\cN_{\text{ren}} +1)$, which can be proved as above, we arrive at  
		\be \label{eq:cHN2}
		\begin{split}
		&\cH_{N}^{(2)}-\cK_{\text{ren}} \\
		& = \frac12 \int dx\, \Big(  N^3V(1-f)(N(x-y)) \ph_t(x)\ph(y) b^*_{x}  a^*(Q_{t,y} )  \frac{a(\ph_t)}{\sqrt{N} }\frac{a(\ph_t)}{\sqrt N}\! +\!\text{h.c.}\Big)\\
		&\hspace{0.3cm}   + \frac N2 \int dx dy \,  N^3 \big( Vf (1-f) - 2 V(1-f) \big) (N(x-y)) |\ph_t(x)|^2 |\ph_t(y)|^2 + \cE_{\cH_{N}^{(2)}} , 
		\end{split}
		\ee
where $\pm \cE_{\cH_{N}^{(2)}}\leq  \delta \cK_{\text{ren}} + \delta^{-1}C  e^{C|t|} (\cN_{\text{ren}} +1)$.
		
In the next step, we extract $\cV_{\text{ren}}$ from $ \cH_{N}^{(4)}$. Here, we first rewrite
		\be\label{eq:aux1}\begin{split}
		a^*(Q_{t,x})a^*(Q_{t,y})& = b^*_{x}a^*(Q_{t,y}) - \frac{a^*(\ph_t)}{\sqrt{N} }\frac{a^*(\ph_t)}{\sqrt N}   a^*(Q_{t,y}) a\big(( Q_t\otimes Q_t k_t)_{x}) \big) \\
		&\hspace{0.5cm} - \frac{a^*(\ph_t)}{\sqrt{N} }\frac{a^*(\ph_t)}{\sqrt N} \, \overline{ Q_t\otimes Q_t k_t}(x,y)  
		\end{split}\ee
Inserting this into $\cH_{N}^{(4)}$ yields with similar arguments as above the decomposition		
		\be \label{eq:cHN4}\begin{split}
		&\cH_{N}^{(4)} -\cV_{\text{ren}} \\
		&=  -\frac12  \int \!dx dy\, N^3 V(1-f) (N(x-y))  \Big( \ph_t(x)\ph_t(y)  b^*_{x}a^*(Q_{t,y})\frac{a(\ph_t)}{\sqrt{N} }\frac{a(\ph_t)}{\sqrt N} +\text{h.c.}\Big)\\
		&\hspace{0.4cm} +   \frac N2 \int dx dy \,  N^3 V(1-f)^2   (N(x-y)) |\ph_t(x)|^2 |\ph_t(y)|^2 + \cE_{\cH_{N}^{(4)}}
		\end{split}\ee	
for an error $  \cE_{\cH_{N}^{(4)}}$ which is controlled by $\pm \cE_{\cH_{N}^{(4)}}\leq  \delta \cV_{\text{ren}} + \delta^{-1}C  e^{C|t|} (\cN_{\text{ren}} +1)$ for every $ \delta>0$ and some constant $C>0$. This follows indeed from the same arguments as above, using Lemma \ref{lm:aux}, Lemma \ref{lm:fock} and Lemma \ref{lm:ker}, together with the simple bound
		\[\begin{split}
		& \Big|  \int dx dy\, N^2 V  (N(x-y))    \ph_t(x)\ph_t(y)   \langle \phi_N, b^*_{x}a^*(Q_{t,y})  a^*\big(( Q_t\otimes Q_t k_t)_{x}) \big)a(Q_{t,y}) \phi_N \rangle\Big|\\
		& \leq  \| \ph_t\|^{3/2}_\infty \langle \cV_{\text{ren}}\rangle_{\phi_N}^{1/2} \Big(  \int dx dy\, N^2 V  (N(x-y))  \langle \phi_N, a^*(Q_{t,y}) \cN_{\bot \ph_t } a(Q_{t,y}) \phi_N \rangle   \Big)^{1/2} \\
		&\leq  C e^{C|t|}  \langle \cV_{\text{ren}}\rangle_{\phi_N}^{1/2}   \langle \,  \cN_{\bot\ph_t}+1\rangle_{\phi_N}^{1/2}  
		\end{split}\]
for every $\phi_N\in L_s^2(\bR^{3N})	$. Finally, inserting \eqref{eq:aux1} into $ \cH_{N}^{(3)}$, we obtain analogously that
		\be \label{eq:cHN3}
		\begin{split}
		\cH_{N}^{(3)} & =   -\Big(   a^*(\ph_t)  a\big (Q_{t} ( N^3 V(1-f)(N.)\ast |\ph_t|^2\ph_t )\big)   + \text{h.c.} \Big) \!+ \! \cE_{\cH_{N}^{(3)}}
		\end{split}\ee
for an error $  \cE_{\cH_{N}^{(3)}}$ controlled by $ \pm \cE_{\cH_{N}^{(3)}}\leq  \delta \cV_{\text{ren}} + \delta^{-1}C  e^{C|t|} (\cN_{\text{ren}} +1) $.

To conclude the proof, it now remains to combine the decompositions in \eqref{eq:cHN2},  \eqref{eq:cHN4},  \eqref{eq:cHN3} with $ \cH_{N}^{(0)}$ and $\cH_{N}^{(1)}$, defined in \eqref{eq:HNdec}. Before doing so, let us observe that 
		\[ \cH_{N}^{(1)} =  \Big( a^*(\ph_t)  a\big (Q_{t} ( N^3 V(1-f)(N.)\ast |\ph_t|^2\ph_t) \big)    +\text{h.c.}\Big) + \cE_{\cH_{N}^{(1)}} \]
for an error $  \cE_{\cH_{N}^{(1)}}$ controlled by $ \pm \cE_{\cH_{N}^{(1)}}\leq C e^{C|t|} (\cN_{\text{ren}} +1) $. This readily follows from the regularity of $ \ph_t$ (see Prop. \ref{prop:GPdyn}) and the identity $\|Vf\|_1 = 8\pi\mathfrak{a}$. Combining this observation with the the identity $a^*(\ph_t) a(\ph_t) = N-\cN_{\bot\ph_t}$, the fact that 
		\[ \langle \ph_t, -\Delta \ph_t\rangle = e_{\text{GP}}   - 4\pi\mathfrak{a} \| \ph_t \|^4_4  = e_{\text{GP}} - \frac12 \big\langle \ph_t, (N^3Vf(N.)\ast |\ph_t|^2)\ph_t\big\rangle +O(1) \]
and the decompositions \eqref{eq:cHN2},  \eqref{eq:cHN4} and \eqref{eq:cHN3}, we conclude that 
		\be\label{eq:cHNdec}\begin{split} 
		\cH_{N} & =    \cH_{\text{ren}}   + \cE_{\cH_{N}}, 
		\end{split} \ee
for an error $\pm  \cE_{\cH_{N}}\leq (1-\delta) \cH_{\text{ren}} + C\delta^{-1} e^{C|t|} (\cN_{\text{ren}} +1) $. Choosing $\delta=\frac12$ concludes \eqref{eq:cHNbnd}.

Let us now switch to the commutator estimate \eqref{eq:cmcHNbnd}. Based on the decomposition of $\cH_N$ in \eqref{eq:cHNdec}, it is useful to split this into two steps and to show separately that
		\be \label{eq:cmbnd1}  \pm  \big[ i \cH_{\text{ren}}, \cN_{\text{ren}}\big] \leq  C\cH_{\text{ren}}  + Ce^{C|t|} (\cN_{\text{ren}} +1) ,    \ee
and
		\be \label{eq:cmbnd2}  \pm  \big[ i \cE_{\cH_{N}}, \cN_{\text{ren}}\big]  \leq  C\cH_{\text{ren}}  + Ce^{C|t|} (\cN_{\text{ren}} +1).   \ee
Proving these bounds requires only a slight variation of the arguments used to derive \eqref{eq:cHNbnd}. We therefore focus on the key ideas for \eqref{eq:cmbnd1} and omit the details for \eqref{eq:cmbnd2}. 		

Let us start with $ [i\cK_\text{ren}, \cN_\text{ren}] $. The key identity that we need is 
		\[\begin{split}
		\big[ b_x, \cN_\text{ren} \big] &=  \big[ a(Q_{t,x}) + N^{-1}a^*\big( (Q_t\otimes Q_t k_t)_x \big) a(\ph_t)a(\ph_t), \cN_{\bot\ph_t} \big] \\
		&\hspace{0.5cm} + N^{-1}\int dz \, \big[ a(Q_{t,x}) ,     a^*\big( (Q_t\otimes Q_tk_t)_z \big) a^*_z a(\ph_t)a(\ph_t)  + \text{h.c.} \big],  \\
		&\hspace{0.5cm} + N^{-2}\int dz \, \big[   a^*\big( (Q_t\otimes Q_t k_t)_x \big) a(\ph_t)a(\ph_t),    a^*(\ph_t)a^*(\ph_t) a\big( (Q_t\otimes Q_tk_t)_z \big) a_z \big] ,  \\
		& = b_x - 2 N^{-2} \int dz\, \big \langle (Q_t\otimes Q_t k_t)_x, (Q_t\otimes Q_t k_t)_z \big \rangle  a_z a^*(\ph_t)a^*(\ph_t) a(\ph_t)a(\ph_t)  \\
		&\hspace{0.5cm} + 2 N^{-2}  a^*\big( (Q_t\otimes Q_t k_t)_x \big) \int dz\,  a\big( (Q_t\otimes Q_tk_t)_z \big) a_z  (2 a^*(\ph_t) a(\ph_t)+1 ). 
		\end{split}\]
Since $\big[ b_x^*, \cN_\text{ren} \big] = - \big[ b_x, \cN_\text{ren} \big]^*$, this implies that $\big[i \cK_\text{ren}, \cN_\text{ren}\big]$ vanishes up to corrections that are quadratic in the kernel $ Q_t\otimes Q_t k_t$. As shown already in the previous step, such correction terms only produce regular terms so that a similar analysis as for \eqref{eq:cHNbnd} implies that $ \pm  [ i \cK_{\text{ren}}, \cN_{\text{ren}}] \leq  C\,\cK_{\text{ren}}  + Ce^{C|t|} (\cN_{\text{ren}} +1)$; we omit the details. The same remarks apply to $ [ i \cV_\text{ren}, \cN_{\text{ren}}]$. Here, we use additionally the identity
		\[\begin{split}
		\big[ a^*(Q_{t,y}) a(Q_{t,y})   ,\cN_\text{ren} \big] &=   2  a^*(Q_{t,y}) a^*\big( (Q_t\otimes Q_t k_t)_y)\big) \frac{a(\ph_t)}{\sqrt{N}} \frac{a(\ph_t)}{\sqrt{N}}    - \text{h.c.} 
		\end{split}\]
Arguing as above, we find
		\[\begin{split}
		\big[  \cV_\text{ren}  ,\cN_\text{ren} \big]&=  \! \int dx dy \,N^2V(N(x-y)) b^*_{x}  \Big( a^*(Q_{t,y}) a^*\big( (Q_t\otimes Q_t k_t)_y)\big) \frac{(a(\ph_t))^2}{N}    - \text{h.c.} \Big) b_x\! \\
		&\hspace{0.5cm}+\!  \cE_{[  \cV_\text{ren}, \,\cN_\text{ren} ]} \\
		\end{split}\]
where $ \pm  i\cE_{[  \cV_\text{ren}, \,\cN_\text{ren} ]}  \leq C\,\cV_{\text{ren}}  + Ce^{C|t|} (\cN_{\text{ren}} +1) $. By Cauchy-Schwarz, we also have
		\[\begin{split} &\pm \Big( \int dx dy \, N^2V(N(x-y)) b^*_{x}   \Big( a^*(Q_{t,y}) a^*\big( (Q_t\otimes Q_t k_t)_y)\big) \frac{a(\ph_t)}{\sqrt{N}} \frac{a(\ph_t)}{\sqrt{N}}   + \text{h.c.} \Big) b_x \\
		&\;\leq  C\,\cV_{\text{ren}}  + Ce^{C|t|} (\cN_{\text{ren}} +1),
		\end{split}\]
so that all in all, we conclude that $\pm [ i \cV_\text{ren}, \cN_{\text{ren}}]\leq C\,\cV_{\text{ren}}  + Ce^{C|t|} (\cN_{\text{ren}} +1)$. 
\end{proof}	
		
The next lemma is our last ingredient needed for the proof of Proposition \ref{prop:aux}. It is similar to the previous Lemma \ref{lm:dec} and collects important properties related to $ \cQ_\text{ren}$. 
\begin{lemma} \label{lm:dec2}
Let $\cH_{N}$ be as in \eqref{def:cHN} and let $ \cQ_\emph{ren} $ be as in \eqref{def:Nren}. Then, for some constant $C>0$ and for every $t\in\bR$, we have that
		\[\begin{split}
		&\pm  \Big(    \big [ i H_N, -  \big(a^*(\ph_t) a(Q_t i\partial_t \ph_t)+\emph{h.c.}\big)  + \langle i\partial_t \ph_t,\ph_t\rangle\cN_{\bot \ph_t} + \cQ_\emph{ren}\big]   \\
		&\hspace{0.8cm} +    \partial_t \big( -  \big(a^*(\ph_t) a(Q_t i\partial_t \ph_t)+\emph{h.c.}\big)  + \langle i\partial_t \ph_t,\ph_t\rangle\cN_{\bot \ph_t} + \cQ_\emph{ren}  \big)\Big)  \\
		&\;\leq C  \cH_\emph{ren}   +  Ce^{C|t|}   (\cN_{\emph{ren}} +1).  
		\end{split}\]
\end{lemma}
\begin{proof}
The proof is based on the same ideas and operator bounds as Lemma \ref{lm:aux} and Lemma \ref{lm:dec}. For this reason, we only outline the key steps. Based on the second bound in \eqref{eq:lmaux2} of Lemma \ref{lm:aux}, we first observe that it suffices to prove that 
		\be\label{eq:lmdec2}\begin{split}
		&\pm  \Big(    \big [ i \cH_N, -  \big(a^*(\ph_t) a(Q_t i\partial_t \ph_t)+\text{h.c.}\big)  + \langle i\partial_t \ph_t,\ph_t\rangle\cN_{\bot \ph_t} + \cQ_\text{ren}\big]   \\
		&\hspace{0.8cm} +    \partial_t \big( -  \big(a^*(\ph_t) a(Q_t i\partial_t \ph_t)+\text{h.c.}\big)  + \langle i\partial_t \ph_t,\ph_t\rangle\cN_{\bot \ph_t}    \big)\Big)  \\
		&\;\leq C      \cH_\text{ren}   +  Ce^{C|t|}   (\cN_\text{ren} +1).  
		\end{split}\ee
Now, we have to compute the contributions on the \rhs explicitly and, in view of \eqref{eq:lmdec2}, it is enough to do this up to errors that are controlled by $ \cH_\text{ren}$ and $\cN_\text{ren}$. We first set   
		\[ \mathcal X  = -  \big(a^*(\ph_t) a(Q_t i\partial_t \ph_t)+\text{h.c.}\big)  + \langle i\partial_t \ph_t,\ph_t\rangle\cN_{\bot \ph_t}, \]  
and find with \eqref{eq:dtQ} that
		\be\label{eq:dtX}\begin{split}
		\partial_t \mathcal X 
		& = \Big(\langle i\partial_t \ph_t,\ph_t\rangle  a^*(\ph_t)a(Q_t\partial_t\ph_t) - a^*(\ph_t)a(Q_t i\partial_t^2\ph_t)  +\text{h.c.}\Big) + \cE_{\partial_t \mathcal X}, 
		\end{split} \ee
where, trivially, $  \pm \partial_t \mathcal X = \pm \big(\partial_t \langle i\partial_t \ph_t,\ph_t\rangle\big)  \cN_{\bot \ph_t}\leq Ce^{C|t| } \cN_\text{ren}  $, by Lemma \ref{lm:aux}. In the next step, a tedious, but straightforward computation shows that 
		\[\begin{split}  
		[i\cH_{N}^{(0)}, \mathcal X ] &=   \Big(  - \big( \| V \|_1 - 8\pi\mathfrak{a}\big)  \|\ph_t\|_4^4  \, a^*(\ph_t) a(Q_t \partial_t \ph_t)+\text{h.c.}\Big)  + \cE_0, \\
		[i\cH_{N}^{(1)}, \mathcal X ]&= N \big( \| V \|_1 - 8\pi\mathfrak{a}\big) \Big( \big\langle  |\ph_t|^2\ph_t, Q_t \partial_t \ph_t\big \rangle +\text{h.c.}\Big) \\
					&\hspace{0.35cm} + \langle \partial_t \ph_t,\ph_t\rangle \big( \| V \|_1 - 8\pi\mathfrak{a}\big)  \Big( a^*(\ph_t) a\big(Q_t  |\ph_t|^2 \ph_t\big) +\text{h.c.}\Big)    + \cE_1,\\
		[i\cH_{N}^{(2)}, \mathcal X ] &=  \Big( a^*(\ph_t) a (Q_t(-\Delta) ( Q_t \partial_t \ph_t) ) + 2 \| V \|_1a^*(\ph_t)  a\big(Q_t  |\ph_t|^2( Q_t \partial_t \ph_t) \big)  \\
					&\hspace{0.7cm}  + \| V \|_1 a^*(\ph_t) a \big( Q_t  \ph_t^2 \overline{ Q_t \partial_t\ph_t}\big) + \text{h.c.}\Big)\\
					&\hspace{0.4cm} - \int \!dx dy\, N^{5/2}V(N(x-y)) \\
					&\hspace{1.2cm}\times\Big( \ph_t (x)\ph_t(y) a^*(Q_{t,x}) a^*(Q_{t,y}) a(Q_t \partial_t \ph_t)  \frac{a(\ph_t)}{\sqrt{N}}\!+\!\text{h.c.} \!\Big)  \\
				  	&\hspace{0.4cm} - \langle \partial_t \ph_t,\ph_t\rangle \int \!dx dy\, N^3V(N(x-y)) \\
				  	& \hspace{2.5cm} \times  \Big( \ph_t (x)\ph_t(y) a^*(Q_{t,x}) a^*(Q_{t,y}) \frac{a(\ph_t)}{\sqrt{N}}\frac{a(\ph_t)}{\sqrt{N}} + \text{h.c.} \Big)   + \cE_2, \\
		[i\cH_{N}^{(3)}, \mathcal X ] &=  \int dx dy \, N^{3} V(N(x-y)) \\
		& \hspace{1cm}\times \Big(  (Q_{t}\partial_t \ph_t)(x)  \ph_t (y) a^*(Q_{t,x}) a^*(Q_{t,y})  \frac{a(\ph_t)}{\sqrt{N}}\frac{a(\ph_t)}{\sqrt{N}} + \text{h.c.}\Big) \\
		&\hspace{0.4cm}  -  \int dx dy \, N^{2} V(N(x-y)) \\
		& \hspace{1cm}\times \Big(    \ph_t (y) a^*(Q_{t,x}) a^*(Q_{t,y}) a(Q_{t,x})  a(Q_t \partial_t \ph_t)  + \text{h.c.}\Big) \\
		&\hspace{0.4cm}-\langle \partial_t \ph_t,\ph_t\rangle \int dx dy \, N^{5/2} V(N(x-y)) \\
		&\hspace{2.4cm}\times \Big( \ph_t (y) a^*(Q_{t,x}) a^*(Q_{t,y})  a(Q_{t,x}) \frac{a(\ph_t)}{\sqrt{N}} + \text{h.c.}\Big)+  \cE_3 , \\
		[i\cH_{N}^{(4)}, \mathcal X ] &=  \int dx dy\, N^{5/2}V(N(x-y))\\
		&\hspace{0.8cm} \times  \Big ( (Q_{t}\partial_t \ph_t)(y)  a^*(Q_{t,x}) a^*(Q_{t,y}) a(Q_{t,x})  \frac{a(\ph_t)}{\sqrt{N} }  +\text{h.c.} \Big) , 
		\end{split} \]
where, for all $j \in \{0,1,2,3\}$, the errors $ \cE_j$ are bounded by
		\[ \pm \cE_j\leq  Ce^{C|t|}(\cN_\text{ren} + 1). \]
This decomposition and the related error bounds are a direct consequence of \eqref{eq:HNdec} and basic estimates as in the proof of the previous lemmas. If we expand $ Q_t = 1- |\ph_t\rangle\langle \ph_t|$ and use that $\text{Re}\, \langle \partial_t\ph_t, \ph_t\rangle=0$, the sum of the different contributions equals 
		\be\label{eq:cHXsimp}\begin{split}
		&[i\cH_{N} , \mathcal X ] \\
		& = N \big( \| V \|_1 - 8\pi\mathfrak{a}\big) \Big( \big\langle  |\ph_t|^2\ph_t, Q_t \partial_t \ph_t\big \rangle +\text{h.c.}\Big) \\
		& \hspace{0.35cm} - \Big(  \langle \partial_t\ph_t, \ph_t\rangle    a^*(\ph_t) a( Q_t \partial_t  \ph_t  ) + \big( \|V\|_1- 8\pi\mathfrak{a}\big) \|\ph_t\|_4^4 \Big( a^*(\ph_t) a( Q_t \partial_t  \ph_t  )  +\text{h.c.}\Big) \\
		& \hspace{0.35cm} + \Big( a^*(\ph_t) a\big( Q_t \partial_t ( -\Delta \ph_t +  \|V\|_1    \,    |\ph_t|^2 \ph_t) \big)   +\text{h.c.}\Big) \\
		& \hspace{0.35cm} +\int dx dy \, N^{2} V(N(x-y))  \Big( \partial_t \ph_t(x)  \ph_t (y) a^*(Q_{t,x}) a^*(Q_{t,y})  a(\ph_t) a(\ph_t)  + \text{h.c.}\Big) \\
		& \hspace{0.35cm}- \int \!dx dy\, N^{2}V(N(x-y)) \Big( \ph_t (x)\ph_t(y) a^*(Q_{t,x}) a^*(Q_{t,y}) a(Q_t \partial_t \ph_t)  a(\ph_t) \!+\!\text{h.c.} \!\Big)  \\
		&\hspace{0.35cm}  -  \int dx dy \, N^{2} V(N(x-y))  \Big(    \ph_t (y) a^*(Q_{t,x}) a^*(Q_{t,y}) a(Q_{t,x})  a(Q_t \partial_t \ph_t)  + \text{h.c.}\Big) \\
		&\hspace{0.35cm}+\int dx dy\, N^{2}V(N(x-y))  \Big (  \partial_t \ph_t(y)  a^*(Q_{t,x}) a^*(Q_{t,y}) a(Q_{t,x})  a(\ph_t)   +\text{h.c.} \Big) +\cE_{[ i\cH_{N} , \mathcal X ]}, 		 
		\end{split}\ee
up to an error $\cE_{[i\cH_{N} , \mathcal X ]}$ that is bounded by $ \cE_{[\cH_{N} , \mathcal X ]}\leq Ce^{C|t|}(\cN_\text{ren} + 1)$. 

Now, observe that the last four terms on the right hand side of \eqref{eq:cHXsimp} are structurally similar to the last contribution to $ \cH_N^{(2)}$ and, respectively, to $\cH_N^{(3)}$, defined in \eqref{eq:HNdec}. We can therefore proceed similarly as in Lemma \ref{lm:dec} and extract their main contributions by comparing them to $ \cH_\text{ren}$. To this end, we substitute \eqref{eq:aux1} into the \rhs of \eqref{eq:cHXsimp} and observe that normal ordering causes cancellations between the terms on the first and fourth lines, the second and fifth lines, and between the terms on the third and last lines. Combined with \eqref{eq:GPdyn} and \eqref{eq:dtX}, we then arrive at
		\be\label{eq:lmdec21} \begin{split}
		&[i\cH_{N} , \mathcal X ] +\partial_t \mathcal X \\
		& =   \big(  \|V\|_1- 8\pi\mathfrak{a}\big)    \,\Big( a^*(\ph_t) a\big( Q_t  \overline{ \ph}_t \partial_t \ph_t^2 \big)  +\text{h.c.}\Big) \\
		&\hspace{0.5cm} +\int dx dy \, N^{2} V(N(x-y))  \Big( \partial_t \ph_t(x)  \ph_t (y) b_x a^*(Q_{t,y})  a(\ph_t) a(\ph_t)  + \text{h.c.}\Big)  +  \cE_{ \mathcal X }, 	
		\end{split}\ee
up to an overall error which is bounded by $ \pm \cE_{\mathcal X}\leq  C \cV_\text{ren} + Ce^{C|t|}   (\cN_{\text{ren}} +1) $. 

It remains to compare the right hand side in \eqref{eq:lmdec21} with $  [ i \cH_N, \cQ_\text{ren}] $. Based on \eqref{eq:HNdec}, a similar computation shows that $ \pm [i \cH_{N}^{(0)}, \cQ_\text{ren} ] \leq Ce^{C|t|}(\cN_\text{ren} + 1)$ and that
		\[\begin{split}  
		[i\cH_{N}^{(1)}, \cQ_\text{ren} ]&=  - \big( \|V\|_1- 8\pi\mathfrak{a} \big)\Big(   a^*(\ph_t ) a \big(  (Q_t\otimes Q_t \partial_t k_t )Q_t |\ph_t|^2\ph_t\big) +\text{h.c.}\Big)   + \Delta_1 ,\\
		[i\cH_{N}^{(2)}, \cQ_\text{ren} ] &= - \frac12 \big(\|V\|_1-8\pi\mathfrak{a}\big) \Big( \big \langle \ph_t^2, \partial_t (\ph^2_t) \big\rangle + \text{h.c.} \Big)\\
		&\hspace{0.5cm}- \int dx\, \Big( a^*(Q_{t,x}) a^* ( (-\Delta_x) ( Q_t\otimes Q_t \partial_t k_t)_x) \frac{a(\ph_t)}{\sqrt{N}}\frac{a(\ph_t)}{\sqrt{N}}  + \text{h.c.} \Big)+ \Delta_2 , \\
		[i\cH_{N}^{(3)}, \cQ_\text{ren} ] &=  \Big( -  \int dx dy \, N^{5/2} V(N(x-y)) \ph_t (y)   \\
		&\hspace{1cm} \times a^*(Q_{t,x}) a^*(Q_{t,y})a^*\big(( Q_t\otimes Q_t \partial_t k_t)_x\big) \frac{a(\ph_t)}{\sqrt{N}}\frac{a(\ph_t)}{\sqrt{N}}\frac{a(\ph_t)}{\sqrt{N}}    \\
		 &\hspace{0.5cm} - \int dx dy \, N^{3} V(N(x-y))  \overline{  \partial_t k_t } (x,y) \ph_t (y) a^*(\ph_t)  a(Q_{t,x})    + \text{h.c.}\Big) + \Delta_3, \\
		[i\cH_{N}^{(4)}, \cQ_\text{ren} ] &=  - \frac12 \Big( \int dx dy\, N^2V(N(x-y))a^*(Q_{t,x}) a^*(Q_{t,y}) \\
		&\hspace{1.1cm}\times \big( \partial_t k_t(x,y)  + 2  a^*\big( ( Q_t\otimes Q_t \partial_t k_t)_x\big) a(Q_{t,y})    \big)   \frac{a(\ph_t)}{\sqrt{N}}\frac{a(\ph_t)}{\sqrt{N}}  + \text{h.c.} \Big) + \Delta_4 , 
		\end{split} \]
up to further error terms $ \Delta_j$ that are controlled by $\pm \Delta_j\leq  Ce^{C|t|}(\cN_\text{ren} + 1)$. 

To combine the different contributions to  $  [ i \cH_N, \cQ_\text{ren}] $, we proceed as before. That is, we substitute \eqref{eq:aux1} and bring all terms into normal order. One then finds that
		\[    [ i \cH_{N}^{(1)} + i \cH_{N}^{(3)}, \cQ_\text{ren} ]  =  -  \big(  \|V\|_1- 8\pi\mathfrak{a}\big)    \,\Big( a^*(\ph_t) a\big( Q_t  \overline{ \ph}_t \partial_t \ph_t^2 \big)  +\text{h.c.}\Big) + \cE_{[i \cH_{N}^{(1)} + i\cH_{N}^{(3)}, \cQ_\text{ren} ]},   \]
where $ \cE_{[i\cH_{N}^{(1)} + i\cH_{N}^{(3)}, \cQ_\text{ren} ]} \leq C \cH_\text{ren} + Ce^{C|t|}   (\cN_{\text{ren}} +1)$. Similarly, based on the zero energy scattering equation \eqref{eq:0en}, the identities \eqref{eq:aux1} and
		\[ \partial_t k_t (x,y)  = N(1-f)(N(x-y) \chi(x-y) \big(  (\partial_t \ph_t) (x) \ph_t (y)+ \ph_t(x) (\partial_t \ph_t)(y)\big) \]
as well as the kernel properties of Lemma \ref{lm:ker}, one readily finds that
		\[\begin{split}
		&[i\cH_{N}^{(2)} + i\cH_{N}^{(4)}, \cQ_\text{ren} ] \\
		& =  - \int dx dy \, N^{2} V(N(x-y))  \Big( \partial_t \ph_t(x)  \ph_t (y) b_x a^*(Q_{t,y})  a(\ph_t) a(\ph_t)  + \text{h.c.}\Big)   \\
		&\hspace{0.5cm}+ \cE_{[i\cH_{N}^{(2)} + i\cH_{N}^{(4)}, \cQ_\text{ren} ]},   
		\end{split}\]
for an error $ \cE_{[i \cH_{N}^{(2)} + i \cH_{N}^{(4)}, \cQ_\text{ren} ]} \leq C \cH_\text{ren} + Ce^{C|t|}   (\cN_{\text{ren}} +1)$. This shows that 
		\be\label{eq:lmdec22}\begin{split}
		&[i\cH_{N} , \cQ_\text{ren} ]   \\
		& =  -  \big(  \|V\|_1- 8\pi\mathfrak{a}\big)    \,\Big( a^*(\ph_t) a\big( Q_t  \overline{ \ph}_t \partial_t \ph_t^2 \big)  +\text{h.c.}\Big) \\
		& \hspace{0.35cm} - \int dx dy \, N^{2} V(N(x-y))  \Big( \partial_t \ph_t(x)  \ph_t (y) b_x a^*(Q_{t,y})  a(\ph_t) a(\ph_t)  + \text{h.c.}\Big)  +   \cE_{ \cQ_\text{ren} } 	
		\end{split}\ee	
for $ \pm \cE_{ \cQ_\text{ren} } \leq  C \cH_\text{ren} + Ce^{Ct}   (\cN_{\text{ren}} +1) $. By direct comparison of \eqref{eq:lmdec21} and \eqref{eq:lmdec22}, we get
		\[ [i\cH_{N} , \mathcal X  ] +\partial_t \mathcal X + [i\cH_{N} , \cQ_\text{ren} ]  =  \cE_{\mathcal X} + \cE_{ \cQ_\text{ren} }. \qedhere \] 
\end{proof}	
		
We conclude this section with the proof of Proposition \ref{prop:aux}. This is now a simple corollary of Lemma \ref{lm:dec} and Lemma \ref{lm:dec2}. 	
\begin{proof}[Proof of Proposition \ref{prop:aux}]
The first bound in Prop.\ \ref{prop:aux} follows directly from Lemma \ref{lm:aux} and Lemma \ref{lm:dec}, so let us focus on the Gronwall bound. Without loss of generality, consider $t\geq 0$. We then compute 
		\be\label{eq:gron1}\begin{split}
		&\partial_t   \big \langle    \cH_N + \cQ_\text{ren}  +  Ce^{Ct}   (\cN_{\text{ren}} +1)  \big\rangle_{\psi_{N,t} } \\
		& =   C^2e^{Ct}  \big \langle   \cN_{\text{ren}} +1  \big\rangle_{\psi_{N,t} } +Ce^{Ct}     \big \langle   \big[i H_N, \cN_{\text{ren}} \big]   + \partial_t \cN_{\text{ren}}   \big\rangle_{\psi_{N,t} }  \\
		&\hspace{0.5cm}+ \big \langle   \big [ iH_N, -  \big(a^*(\ph_t) a(Q_t i\partial_t \ph_t)+\text{h.c.}\big)  + \langle i\partial_t \ph_t,\ph_t\rangle\cN_{\bot \ph_t} + \cQ_\text{ren}\big]   \big\rangle_{\psi_{N,t} } \\
		&\hspace{0.5cm} + \big \langle  \partial_t \big( -  \big(a^*(\ph_t) a(Q_t i\partial_t \ph_t)+\text{h.c.}\big)  + \langle i\partial_t \ph_t,\ph_t\rangle\cN_{\bot \ph_t} +\cQ_\text{ren}\big)  \big\rangle_{\psi_{N,t} }.  
		\end{split}\ee
By \eqref{def:cHN} and Lemma \ref{lm:dec}, the second term on the \rhs in \eqref{eq:gron1} is controlled by 
		\[\begin{split}
		&\big|  \big \langle   \big[ i H_N, \cN_{\text{ren}} \big]   + \partial_t \cN_{\text{ren}}   \big\rangle_{\psi_{N,t} } -  \big \langle   \big[i \cH_N,  \cN_{\text{ren}} \big] \big\rangle_{\psi_{N,t} } \big|\\
		&\hspace{0.5cm}\leq ce^{ct}\big \langle\cH_{\text{ren}} + C e^{Ct} (\cN_{\text{ren}} +1)\big\rangle_{\psi_{N,t} } \leq ce^{ct} \big \langle    \cH_N +\cQ_\text{ren} +  Ce^{Ct}   (\cN_{\text{ren}} +1)  \big\rangle_{\psi_{N,t} }
		\end{split}\]
and, based on the same lemma and \eqref{eq:NNren}, we also obtain that 
		\[\begin{split}
		\big|    \big \langle   \big[i \cH_N,  \cN_{\text{ren}} \big] \big\rangle_{\psi_{N,t} } \big| \leq ce^{ct} \big \langle    \cH_N +\cQ_\text{ren}  +  Ce^{Ct}   (\cN_{\text{ren}} +1)  \big\rangle_{\psi_{N,t} }. 
		\end{split}\]
Similarly, Lemma \ref{lm:dec2} implies directly that
		\[\begin{split}
		& \Big| \big \langle   \big [ i H_N, -  \big(a^*(\ph_t) a(Q_t i\partial_t \ph_t)+\text{h.c.}\big)  + \langle i\partial_t \ph_t,\ph_t\rangle\cN_{\bot \ph_t} + \cQ_\text{ren}\big]   \big\rangle_{\psi_{N,t} } \\
		&\hspace{0.2cm} + \big \langle  \partial_t \big( -  \big(a^*(\ph_t) a(Q_t i\partial_t \ph_t)+\text{h.c.}\big)  + \langle i\partial_t \ph_t,\ph_t\rangle\cN_{\bot \ph_t}+ \cQ_\text{ren}\big)  \big\rangle_{\psi_{N,t} }\Big|  \\
		&\leq ce^{ct} \big \langle    \cH_N  + \cQ_\text{ren}+  Ce^{Ct}   (\cN_{\text{ren}} +1)  \big\rangle_{\psi_{N,t} }.  \qedhere
		\end{split}\]
\end{proof}

\appendix 

\section{Properties of the Gross-Pitaevskii Equation} \label{app:dyn}

The next proposition collects basic properties of the solution of the time-dependent Gross-Pitaevskii equation \eqref{eq:GPdyn}. Its proof follows essentially from standard arguments; we refer to \cite[Appendix A]{BDS} and \cite[Chapters 3 to 6]{Caz} for the details.  
\begin{prop}\label{prop:GPdyn}
Consider the time dependent Gross-Pitaevskii equation \eqref{eq:GPdyn}. Then: 
\begin{enumerate}[a)]

	\item  \textsc{Well-Posedness.} For every $\ph \in H^1(\bR^3)$ with $\| \ph \|_2 = 1$, there exists a unique global solution $t \to \ph_t \in C(\bR, H^1(\bR^3))$ of \eqref{eq:GPdyn} with initial data $ \ph$. We have that $\| \ph_t \|_2  = 1$ and that $ \cE_{\emph{GP}}(\ph_t) = \cE_{\emph{GP}}(\ph)$ for all $t \in \bR$. In particular, we have that
		\[ \sup_{t\in\bR} \| \ph_t \|_{H^1}   \leq C, \hspace{0.5cm} \sup_{t\in\bR} \| \ph_t \|_{4} \leq C.  \]
	
	\item \textsc{Higher Regularity.} If $ \varphi\in H^m(\bR^3)$ for some $m \geq 2$, then $\ph_t \in H^m(\bR^3)$ for every $t\in\bR$. Moreover, there exists $C>0$, depending on $m$ and on $\|\varphi\|_{H^m}$, and $c>0$, depending on $m$ and $\| \varphi\|_{H^1}$, such that for all $t\in \bR$, we have 
	   \[ \| \ph_t \|_{H^m} \leq Ce^{c|t|}. \]
	
	\item  \textsc{Regularity of Time Derivatives.} If $\varphi\in H^4(\bR^3)$, there exists $C>0$, depending on $\| \varphi\|_{H^4}$, and $c>0$, depending on $ \| \varphi\|_{H^1}$, such that for all $t \in \bR$, we have that
		\[ \|  \partial_t \ph_t \|_{H^2} \leq Ce^{c|t|},\;\; \| \partial_t^2 \ph_t \|_{H^2} \leq Ce^{c|t|}. \]	
\end{enumerate}
\end{prop}

\section{Basic Fock Space Operators} \label{app:fock}

In this section, we collect a few standard results on the creation and annihilation operators defined in Section \ref{sec:intro}. The proof of the following lemma is straightforward and follows with the same arguments as in \eg \cite[Section 2]{BDS} or \cite[Section 2]{BS}. 

\begin{lemma}\label{lm:fock} 
Let $ f, g \in L^2(\bR^{3}), h \in L^2(\bR^3\times \bR^3)$ and let the $ a_x, a^*_y$ and $ a(Q_{t,x}), a^*(Q_{t,y})$ be defined as in \eqref{def:aas} and \eqref{def:aQ}, respectively. Then, in $L^2_s(\bR^{3N})$ we have that
		\[  0\leq \cN_{\bot \ph_t} = N - a^*(\ph_t)a(\ph_t) = \int dx\, a^*(Q_{t,x}) a(Q_{t,x}) \leq \int dx\, a^*_x a_x =  N. \]
Moreover, for every $ \phi_N\in L^2_s(\bR^{3N})$, we have that 
		\[\begin{split}
		\| a(f)\phi_N\| & \leq \| f\| \sqrt N \|\phi_N\|, \;\;\; \| a^*(f)\phi_N\| \leq \| f\| \sqrt {N+1} \|\phi_N\|, \\
		\| a(Q_t f)\phi_N\| & \leq \| f\|  \| \cN_{\bot \ph_t}^{1/2}\phi_N \|, \;\;\; \| a^*(Q_t f)\phi_N\| \leq \| f\| \| (\cN_{\bot \ph_t}+1)^{1/2}\phi_N\|, \\
		\end{split}\] 
and that 
		\[\big|  \langle a^*(f) a(g)\rangle_{\phi_N} \big| \leq N \| f\| \| g\|\|\phi_N\|^2 , \;  \big|  \langle a^*(Q_tf) a(Q_tg)\rangle_{\phi_N} \big|\leq \langle  \| Q_t f\| \| Q_tg\|\cN_{\bot \ph_t}\rangle_{\phi_N}.\]
Similarly, if we set $ h_{x}(y) = h(x,y)$, then we have that
		\[\begin{split} 
	 	\int dx \, \big| \big\langle a^*_x a^* (h_x) \big \rangle_{\phi_N} \big| & \leq N \|h \| \|\phi_N\|^2, \int dx \, \big| \big\langle a^*(Q_{t,x}) a^* (Q_t h_x) \big \rangle_{\phi_N} \big|  \leq  \| h \| \langle \cN_{\bot \ph_t}\rangle_{\phi_N}\\
		\int dx \, \big| \big\langle a^*_x a (h_x) \big \rangle_{\phi_N} \big| &\leq N \|h \| \|\phi_N\|^2, \int dx \, \big| \big\langle a^*(Q_{t,x}) a (Q_t h_x) \big \rangle_{\phi_N} \big|   \leq  \|h \|  \langle \cN_{\bot \ph_t}\rangle_{\phi_N}.
		\end{split} \]
\end{lemma}

\section{Properties of the Scattering Kernel} \label{app:kern}

The goal of this section is to collect basic properties of the solution $f$ of the zero energy scattering equation \eqref{eq:0en} and of the correlation kernel $k_t$, defined in \eqref{def:k}. It is well known (see \cite[Appendix C]{LSSY}) that under our assumptions on $V\in L^1(\bR^3)$, we have that $0 \leq f \leq 1$, that $f$ is radially symmetric and radially increasing and that for every $x\in \bR^3$ outside of the support of $V$, it holds true that
		\be \label{eq:flrgx}  f(x) = 1- \frac{\mathfrak{a}}{|x|}. \ee 
In particular, we have that $ 0\leq (1-f)(x) \leq C|x|^{-1}  $ and, if $\text{supp}(V)\subset B_R(0)$, that
		\[ \int_{\bR^3} dx\, V(x) f(x) =  2 \int_{B_R(0) } dx\,  (\Delta f)(x) =  2 \int_{ \partial B_R(0) } dS(x)  \, (\nabla f)(x) \cdot \frac{x}{|x|} = 8\pi \mathfrak{a}.   \]

\begin{lemma}\label{lm:ker}
Let $k_t$ be defined as in \eqref{def:k}, where $ t\mapsto  \ph_t \in C^1(\bR, H^1(\bR^3))$ denotes the unique solution of the time-dependent Gross-Pitaevskii equation with $\ph_{t=0}\in H^4(\bR^3)$ and where $\chi \in C_c^\infty (B_{2r}(0))$ with $ \chi_{|  B_{r}(0)}\equiv 1$. Then, $k_t$ satisfies the following properties: 
\begin{enumerate}[a)]

\item We have that $ k_t \in L^2(\bR^3\times \bR^3)$ with 
		\[ \sup_{t\in\bR }\| k_t\| \leq Cr^{1/2}\;\; \text{ and } \;\;\| k_{t,x} \| \leq C |\ph_t(x)| \leq Ce^{C|t|}  \]
for some constant $C>0$ that is independent of $r>0$ and $t\in \bR$. Similarly, we have  
		\[\begin{split}
		\| \partial_t k_t \| &\leq C e^{C|t|} \;\; \text{ and } \;\;\| \partial_t k_{t,x} \| \leq C \big( |\ph_t(x)| + |\partial_t \ph_t(x)| \big) \leq Ce^{C|t|}, \\
		 \| \partial_t^2 k_t \| &\leq C e^{C|t|} \;\; \text{ and } \;\;\| \partial_t^2 k_{t,x} \| \leq C \big( |\ph_t(x)| + |\partial_t \ph_t(x)|+ |\partial_t^2 \ph_t(x)| \big) \leq Ce^{C|t|}. \\
		\end{split}\]
The same bounds hold true if we replace $ k_t $ by $Q_t\otimes Q_t k_t$ for $Q_t = 1- |\ph_t\rangle\langle\ph_t|$.
\item  Define $ f_t(x,y) $ by 
		\[\begin{split}
		f_t(x,y) & = (-\Delta_1 k_t) (x,y)   - \frac 12 N^3(Vf)(N(x-y)) \ph_t(x)\ph_t(y) \\
		&\hspace{0.5cm} - 2 N^2 (\nabla f)(N(x-y)) \cdot \nabla_1\big(\chi(x-y) \ph_t(x) \ph_t(y)\big).
		\end{split}\]
Then $ f_t, \partial_t f_t  \in L^2(\bR^3\times \bR^3)$ with
 		\[ \sup_{t\in\bR } \| f_t\| \leq C, \;\;  \| \partial_t f_t\| \leq  C e^{C|t|}.  \]
\item Define $ g_t(x,y) = \int dz \, (\nabla_2 k_t)(x,z) (\nabla_2 \overline k_t)(y,z)  $. Then $ g_t, \partial_t g_t  \in L^2(\bR^3\times \bR^3)$ with 
		\[    \| g_t\| \leq  C e^{C|t|}, \;\;  \| \partial_t g_t\| \leq  C e^{C|t|}.  \]
The same bounds hold true if in the definition of $ g_t$ we replace $ k_t $ by $Q_t\otimes Q_t k_t$.

\end{enumerate}
\end{lemma}
\begin{proof}
We sketch the main steps of the proof for the bounds involving $k_t$; similar properties have previously been used in \cite{BDS,BS}. Below, we use without further notice that $ \| \ph_t \|_\infty, \| \partial_t \ph_t \|_\infty, \| \partial_t^2\ph_t \|_\infty\leq Ce^{C|t|}   $ and that $ \|\ph_t\| =\|\ph_{t=0}\|$, $ \sup_{t\in\bR} \| \nabla \ph_t \|\leq C,   \sup_{t\in\bR}\|  \ph_t \|_4\leq C  $ by Prop. \ref{prop:GPdyn} and $ H^2(\bR^3)\hookrightarrow L^\infty(\bR^3) $. 

Part $a)$ follows from  
		\[ \|k_t\|^2 \leq \int dxdy\, \frac{ |\chi(x-y)|^2 }{|x-y|^2} |\ph_t(x)|^2|\ph_t(y)|^2 \leq C \|\ph_t\|_4^4 \int_{|x|\leq 2r }\frac{1}{|x|^2}\leq C r  \]
uniformly in $t\in\bR$ and, using Hardy's inequality, from
		\[\|k_{t,x}\|^2 \leq |\ph_t(x)|^2 \int dy\, \frac{ |\chi(x-y)|^2 }{|x-y|^2} |\ph_t(y)|^2 \leq C |\ph_t(x)|^2 \| \nabla \ph_t (\cdot +x)\| \leq C |\ph_t(x)|^2.   \]
The bounds on the time derivatives of $ k_t$ are proved in the same way. This remark applies also to part $b)$ which follows after noting that 
		\[\begin{split}
		f_t(x,y) & =   N(1-f)(N(x-y)) (-\Delta_1)  \big(\chi(x-y) \ph_t(x) \ph_t(y) \big).
		\end{split}\]
This uses the zero energy scattering equation and that $N$ is w.l.o.g.\  large enough so that $ V(N.) \chi(.)\equiv V(N.)$. Finally, to prove part $c)$, we compute 
		\[\begin{split}
		(\nabla_2 k_t)(x,z) &=  N^2 \nabla f (N(x-z)) \chi(x-z)\ph_t(x)\ph_t(z) \\
		&\hspace{0.5cm}+ N(1-f)(N(x-z)) \nabla_2\big( \chi(x-z) \ph_t(x)\ph_t(z)\big).
		\end{split}\]
Setting $ h(x,z) = \nabla_2\big( \chi(x-z) \ph_t(x)\ph_t(z)\big)$, we then note by Hardy's inequality that 
		\[\begin{split}
		&\int dxdy\,\bigg| \int  dz \,N   f (N(x-z))N   f(N(y-z)) h(x,z) h(y,z) \bigg|^2  \\
		&\leq \! C \! \int  dxdy dz_1 dz_2 \,    \frac{ \prod_{j=1}^2| | \ph_t(z_j)|^2+ | \nabla \ph_t(z_j)|^2 |  }{|x-z_1||y-z_1||x-z_2||y-z_2|}| \ph_t(x)|^2 | \ph_t(y)|^2\leq C e^{C|t|}.
		\end{split}\]
and, similarly, that 
		\[\begin{split}
		&\int dxdy\,\bigg| \int  dz \,N^2   \nabla f (N(x-z)) N   f(N(y-z))\chi(x-z)\ph_t(x)\ph_t(z) h(y,z) \bigg|^2  \\
		&\leq  C \int   dxdy  du_1du_2 dz_1 dz_2 \,   N^3V(Nu_1)N^3V(Nu_2)  \\
		&\hspace{1cm}\times \frac{ | \ph_t(x)|^2  | \ph_t(y)|^2 \ph_t(z_1)|^2  || \ph_t(z_2)|^2+ | \nabla \ph_t(z_2)|^2||}{|x-z_1-u_1|^2 |y-z_1| |x-z_2-u_2|^2|y-z_2|}\\
		&\leq C e^{C|t|}.
		\end{split}\]
Here, we used that $N^2 \nabla f (Nx) = -\frac1{8\pi } \int dy\,  \frac{x-y}{|x-y|^3} N^3V(Ny) $ for $a.e.$ $x\in\bR^3$, which follows from $ (-2 \Delta  + N^2 V(N.))f(N.) =0$. Finally, by integration by parts, we use that
		\[\begin{split}
		&\int  dz \,N^2   \nabla f (N(x-z)) N^2   \nabla f(N(y-z))  \chi(x-z)\ph_t(x) \chi(y-z)\ph_t(y) | \ph_t(z)|^2\\
		& = \frac12  \int  dz \,N   f (N(x-z)) N^3 (Vf)(N(y-z))  \chi(x-z)\ph_t(x) \chi(y-z)\ph_t(y) | \ph_t(z)|^2 \\
		&\hspace{0.35cm} + \int  dz \,N   f (N(x-z)) N^2   \nabla f(N(y-z))  \nabla_z\big(   \chi(x-z)\ph_t(x) \chi(y-z)\ph_t(y) | \ph_t(z)|^2\big).
		\end{split}\]
Then, proceeding as before, we find that
		\[\begin{split}
		&\int dxdy\,\Big|  \int  dz \,N   f (N(x-z)) N^3 (Vf)(N(y-z))  \chi(x-z)\ph_t(x) \chi(y-z)\ph_t(y) | \ph_t(z)|^2\Big|^2\\
		 &\leq C  \int dx dy dz_1 dz_2  \, N^3 V(N(y-z_1))N^3 V(N(y-z_2)) \frac{|\ph_t(x)|^2|\ph_t(y)|^2| \ph_t(z_1)|^2| \ph_t(z_2)|^2 }{|x-z_1| |x-z_2|}\\
		 &\leq C e^{C|t|}  \int dx dy dz_1 dz_2  \, N^3 V(N(y-z_1))N^3 V(N(y-z_2)) \frac{|\ph_t(x)|^2|\ph_t(y)|^2  }{|x-z_1|^2 }\leq C e^{C|t|}
		 \end{split}\]
and that
		\[\begin{split}
		&\int dxdy\,\Big| \!\int \!  dz \,N^3  f (N(x-z))    \nabla f(N(y-z))  \nabla_z\big(   \chi(x-z)\ph_t(x) \chi(y-z)\ph_t(y) | \ph_t(z)|^2\big)\Big|^2\\
		 &\leq C  \int dx dy dz_1 dz_2 du_1 du_2  \, N^3 V(Nu_1 )N^3 V(N u_2 ) \frac{|\ph_t(x)|^2|\ph_t(y)|^2    }{|x-z_1| |x-z_2|}\\
		 &\hspace{1cm} \times \frac{ \big ( |  \ph_t(z_1)|^2 + | \nabla \ph_t(z_1)|^2\big)\big ( |  \ph_t(z_2)|^2 + | \nabla \ph_t(z_2)|^2\big) }{|y-u_1-z_1|^2 |y-u_2-z_2|^2} \\
		 &\leq C e^{C|t|}  
		 \end{split}\]
Combining the above, this implies the bounds on $g_t$ and $\partial_t g_t$ is bounded similarly.
\end{proof}

\section{Complete BEC for Small Interaction Potentials} \label{app:smallV}

The purpose of this appendix is to illustrate that the methods developed in this paper are also useful in the spectral setting. The following result generalizes the main result of \cite{BBCS1} to the trapped setting in $\bR^3$; compared to \cite{BBCS1}, its proof is substantially simpler. 

\begin{prop}\label{prop:BEClwr} 
Let $ V\in L^1(\bR^3)$ be non-negative, radial, compactly supported and such that $\|V\|_1$ is small enough. Let $V_\emph{ext}\in L_{\emph{loc}}^\infty(\bR^3)  $ be such that $ \lim_{|x|\to\infty} V_{\emph{ext}}(x) =\infty$ with at most exponential growth in $|x|$ as $|x|\to\infty$. Then, there exists a constant $C>0$, that only depends on $V$, such that for every $ \psi_N$, $ \|\psi_N\|=1$, that satisfies
		\[     \langle \psi_N, H_N^\emph{trap} \psi_N\rangle \leq N e_\emph{GP}^\emph{trap} + \Lambda, \]
we have that the one particle density $ \gamma_N^{(1)} $ associated to $\psi_N$ satisfies 
		\[ 1- \langle \ph_\emph{GP}, \gamma_N^{(1)}\ph_\emph{GP}\rangle\leq   C N^{-1}( 1+ \Lambda) .     \]
\end{prop}
\noindent \textbf{Remarks:} 
\begin{enumerate}
\item  [D1)] It is well-known \cite{LSY} that $ \inf\text{spec}(H_N^\text{trap}) = N e_\text{GP}^\text{trap} + o(1) $ as $N\to\infty$. In particular, Proposition \ref{prop:BEClwr} applies to the ground state $\psi_N$ of $H_N^\text{trap}$. 
\item  [D2)] Based on \cite{LSY}, it is well-known that under our assumptions we have that $\ph_\text{GP}  $ decays exponentially fast to zero as $|x|\to\infty$, with arbitrary rate. In particular, we have that $ V_\text{ext} \ph_\text{GP} \in L^p(\bR^3)$ for every $ p\geq 1$. 
\end{enumerate}
\vspace{0cm}
\begin{proof}[Proof of Prop.\ \ref{prop:BEClwr}]
Using the Euler-Lagrange equation for $\ph_\text{GP}$, that is
		\[  -\Delta + V_{\text{ext}} + 8\pi\mathfrak{a}|\ph_\text{GP}|^2 \ph_\text{GP} =   \mu_\text{GP} \ph_\text{GP}, \hspace{0.5cm}  \mu_\text{GP}= e_\text{GP}+ 4\pi\mathfrak{a}\|\ph_{\text{GP}}\|_4^4,  \]
the proof follows from a slight variation of the arguments presented in Section \ref{sec:aux}. Indeed, proceeding as in \eqref{eq:HNdec}, it is straightforward to verify that $ H_N^\text{trap} = \sum_{j=0}^4 H_{N,j}^\text{trap}$, where
		\[\begin{split} \label{eq:HNtrapdec}
		H_{N,0}^\text{trap}&= \frac N2  \big\langle \ph_{\text{GP}}, (N^3V(N.) \ast|\ph_{\text{GP}}|^2) \ph_{\text{GP}} \big\rangle  \frac{a^*(\ph_{\text{GP}})}{\sqrt{N}}\frac{a^*(\ph_{\text{GP}})}{\sqrt{N}} \frac{a(\ph_{\text{GP}})}{\sqrt{N}}\frac{a(\ph_{\text{GP}})}{\sqrt{N}}   \\
		&\hspace{0.4cm} +  N \big\langle \ph_{\text{GP}}, (-\Delta +  V_\text{ext}) \ph_{\text{GP}} \big\rangle \frac{a^*(\ph_{\text{GP}})}{\sqrt{N}} \frac{a(\ph_{\text{GP}})}{\sqrt{N}}   \\
		H_{N,1}^\text{trap}&=  a^*(\ph_{\text{GP}}) a \big(Q_t ( N^3V (N.) \ast|\ph_{\text{GP}}|^2) \ph_{\text{GP}}) \big)-   a^*(\ph_{\text{GP}}) a \big(Q_t  ( 8\pi \mathfrak{a} |\ph_{\text{GP}}|^2 \ph_{\text{GP}}) \big)\\
		&\hspace{0.5cm} -  \frac{a^*(\ph_{\text{GP}})}{\sqrt{N}} a \big(Q_t ( N^3V(N.) \ast|\ph_{\text{GP}}|^2) \ph_{\text{GP}}) \big) \frac{\cN_{\bot \ph_{\text{GP}}} }{\sqrt{N}} +\text{h.c.}, \\
		H_{N,2}^\text{trap} &= \int dx\, a^*(Q_{x}) (-\Delta_x + V_\text{ext}(x)) a(Q_{x}) \\
				  &\hspace{0.3cm} +\int dx dy  \, N^3 V(N(x-y)) |\ph_{\text{GP}} (y)|^2   a^*(Q_{x}) a (Q_{x}) \frac{a^*(\ph_{\text{GP}})}{\sqrt{N}}\frac{a(\ph_{\text{GP}})}{\sqrt{N}} \\
				  &\hspace{0.3cm}+\int dx dy  \, N^3 V(N(x-y)) \ph_{\text{GP}} (x) \overline{\ph}_t (y) a^*(Q_{x}) a(Q_{y})  \frac{a^*(\ph_{\text{GP}})}{\sqrt{N}}\frac{a(\ph_{\text{GP}})}{\sqrt{N}} \\
    				  &\hspace{0.3cm} +\!   \frac12\int \!dx dy\, N^3V(N(x-y)) \!\Big( \ph_{\text{GP}} (x)\ph_{\text{GP}}(y) a^*(Q_{x}) a^*(Q_{y}) \frac{a(\ph_{\text{GP}})}{\sqrt{N}}\frac{a(\ph_{\text{GP}})}{\sqrt{N}}\!+\!\text{h.c.} \!\Big), \\
		H_{N,3}^\text{trap} &= \int dx dy \, N^{5/2} V(N(x-y)) \ph_{\text{GP}} (y) a^*(Q_{x}) a^*(Q_{y})  a(Q_{x}) \frac{a(\ph_{\text{GP}})}{\sqrt{N}} + \text{h.c.} , \\
		H_{N,4}^\text{trap}&= \frac1{2}\int dx dy\, N^2V(N(x-y)) a^*(Q_{x}) a^*(Q_{y}) a(Q_{y}) a(Q_{x}) . 
		\end{split} \]
Here, we set $ Q_{x} = (Q_{t,x})_{|t=0}$ compared to our previous notation. Similarly, we continue to use the notation $ b_x, b^*_y, \cN_\text{ren}, \cH_\text{ren}, k \equiv (k_t)_{|t=0}$, etc., understanding implicitly that this refers to $t=0$ so that all operators are related to $\ph_\text{GP}$. Now, to control $H_N^\text{trap}$ relative to $ \cH_\text{ren}$ and $\cN_\text{ren}$, a simple generalization of the arguments in Section \ref{sec:aux} shows that
		\[\begin{split} 
		H_N^\text{trap} & \geq  \frac{N} 2 \big\langle \ph_{\text{GP}}, N^3(Vf)(N.)\ast | \ph_{\text{GP}} |^2\ph_{\text{GP}} \big\rangle  \frac{a^*(\ph_{\text{GP}})}{\sqrt{N}}\frac{a^*(\ph_{\text{GP}})}{\sqrt{N}} \frac{a(\ph_{\text{GP}})}{\sqrt{N}}\frac{a(\ph_{\text{GP}})}{\sqrt{N}}\\
		&\hspace{0.3cm} +  N \big\langle \ph_{\text{GP}}, (-\Delta +  V_\text{ext}) \ph_{\text{GP}} \big\rangle \frac{a^*(\ph_{\text{GP}})}{\sqrt{N}} \frac{a(\ph_{\text{GP}})}{\sqrt{N}}  \\
		&\hspace{0.3cm}  + \frac12 \int dx \, b^*_x  \Big(-\Delta_x + V_\text{ext}(x) + N^3V(N.)\ast | \ph_{\text{GP}} |^2 \ph_{\text{GP}}   \Big)  b_x   - C \|V\|_1 (\cN_{\bot \ph_\text{GP}}+1 ) 
		\end{split} \] 
for some universal $C>0$. Notice that this uses $ \cV_\text{ren}\geq 0$. Using the regularity of $\ph_\text{GP}\in H^1(\bR^3)$, the property $0\leq f\leq 1$, the identity $a^*(\ph_{\text{GP}})a(\ph_{\text{GP}}) = N- \cN_{\bot \ph_\text{GP}} $ and that $ \cN_{\bot \ph_\text{GP}}$ and $\cN_\text{ren}$ are of comparable size by Lemma \ref{lm:aux}, we get  
		\[\begin{split} 
		&H_N^\text{trap} -Ne_\text{GP}^\text{trap}\\
		& \geq      \frac12 \int dx \, b^*_x  \big(-\Delta_x + V_\text{ext}(x) +  8\pi\mathfrak{a} | \ph_{\text{GP}} |^2 \ph_{\text{GP}}(x)  -\mu_\text{GP} \big)  b_x   - C  \|V\|_1(\cN_{\bot \ph_\text{GP}}\!+1 ).
		\end{split} \] 
Here, we chose the radius $r>0$ in the definition of \eqref{def:k} w.l.o.g.\ comparable to $\|V\|_1$. Finally, standard results imply that $h_\text{GP} =  -\Delta  + V_\text{ext}(x) +  8\pi\mathfrak{a} | \ph_{\text{GP}} |^2 \ph_{\text{GP}}  -\mu_\text{GP} $ is gapped above its ground state energy, for some gap $2 \lambda_\text{GP}>0$. By the Euler-Lagrange equation, $ \ph_\text{GP}$ is its unique positive ground state (with zero ground state energy) so that 
		\[\begin{split} 
		H_N^\text{trap}  \geq  Ne_\text{GP}^\text{trap}   +  \lambda_\text{GP} \,  \cN_\text{ren}    - C \|V\|_1  (\cN_{\bot \ph_\text{GP}}+1 ) \geq Ne_\text{GP}^\text{trap}  + \delta\,  \cN_{\bot \ph_\text{GP} } - C \|V\|_1- C,
		\end{split} \] 
for  $   \delta=    \lambda_\text{GP}-C \|V\|_1>0 $, if $\|V\|_1$ is small enough. By \eqref{eq:linkNbot}, this implies the claim. 
\end{proof}


\medskip
\noindent\textbf{Acknowledgements.} C. B. and W. K. acknowledge support by the Deutsche Forschungsgemeinschaft (DFG, German Research Foundation) under Germany’s Excellence Strategy – GZ 2047/1, Projekt-ID 390685813.


\end{document}